\documentclass[11pt,letterpaper]{article}

\usepackage{amsmath,amsthm,amsfonts,amssymb}
\usepackage{thm-restate}
\usepackage{fullpage}
\usepackage[utf8]{inputenc}
\usepackage[dvipsnames]{xcolor}
\usepackage{xspace,enumerate}
\usepackage[shortlabels]{enumitem}
\usepackage[hypertexnames=false,colorlinks=true,urlcolor=Blue,citecolor=Green,linkcolor=BrickRed]{hyperref}
\usepackage[capitalise]{cleveref}
\usepackage[OT4]{fontenc}
\usepackage[ruled]{algorithm}
\usepackage[noend]{algorithmic}
\usepackage{ifpdf}
\usepackage{todonotes}
\usepackage{enumerate}
\usepackage{authblk}
\usepackage{cite}
\usepackage{thmtools}

\def\poly{\operatorname{poly}}

\title{Fully Dynamic Algorithms for Minimum Weight Cycle\\and Related Problems}
\date{\vspace{-5ex}}
\author[1]{Adam Karczmarz\thanks{\texttt{a.karczmarz@mimuw.edu.pl}. Supported by ERC Consolidator
Grant 772346 TUgbOAT and by the Foundation for Polish Science (FNP) via the START programme.}}
\affil[1]{Institute of Informatics, University of Warsaw, Poland}

\newcommand{\Ot}{\ensuremath{\widetilde{O}}}
\newcommand{\eps}{\ensuremath{\epsilon}}
\newcommand{\dist}{\delta}
\newcommand{\len}{\ell}
\newcommand{\wei}{w}

\theoremstyle{plain}
\newtheorem{theorem}{Theorem}
\newtheorem{lemma}[theorem]{Lemma}
\newtheorem{corollary}[theorem]{Corollary}

\newtheorem{fact}[theorem]{Fact}
\newtheorem{observation}[theorem]{Observation}

\newtheorem{remark}[theorem]{Remark}

\begin{document}

\maketitle

\begin{abstract}
  We consider the directed minimum weight cycle problem in the fully dynamic setting.
  To the best of our knowledge, so far no fully dynamic algorithms have been designed
  specifically for the minimum weight cycle problem in general digraphs. One can achieve $\Ot(n^2)$
  amortized update time by simply invoking
  the fully dynamic APSP algorithm of Demetrescu and Italiano~[J. ACM~'04].
  This bound, however, yields no improvement over the trivial recompute-from-scratch 
  algorithm for sparse graphs.

  Our first contribution is a very simple deterministic $(1+\eps)$-approximate algorithm supporting
  vertex updates (i.e., changing all edges incident to a specified vertex)
  in conditionally near-optimal $\Ot(m\log{(W)}/\eps)$ amortized time for digraphs
  with real edge weights in $[1,W]$. 
  Using known techniques, the algorithm can be implemented 
  on planar graphs and also gives some new sublinear
  fully dynamic algorithms maintaining approximate cuts and flows in planar digraphs.

  Additionally, we show a Monte Carlo randomized exact fully dynamic minimum weight cycle 
  algorithm with $\Ot(mn^{2/3})$ worst-case update
  that works for real edge weights.
  To this end, we generalize the exact fully dynamic APSP data structure of Abraham et al. [SODA'17] to
  solve the \emph{multiple-pairs shortest paths} problem, where one is interested in
  computing distances for some~$k$ (instead of all $n^2$) fixed source-target pairs after each update.
  We show that in such a scenario, $\Ot((m+k)n^{2/3})$ worst-case update time is possible.
\end{abstract}

\section{Introduction}
The all-pairs shortest paths problem (APSP) is one of the most fundamental graph problems.
Given a real-weighted \emph{directed} graph $G$ with $n$ vertices, the goal is to
compute the distance matrix between all pairs of vertices $u,v$ in $G$.
APSP can be computed in $\Ot(nm)$ time~\cite{Johnson77, Pettie04}, which is clearly near-optimal
for sparse graphs (since the output consists of $n^2$ numbers), but is
also conjectured to be optimal for the entire range of possible graph sparsities.
Some of the other core directed graph problems such as computing the diameter, the radius, or the minimum weight
cycle\footnote{Also called the \emph{girth}, or the \emph{weighted girth} of a digraph. For simplicity,
in this paper we very often use \emph{minimum weight cycle} to refer to the \emph{length}
of such a cycle rather than to the actual cycle. Moreover, throughout this paper, our focus is on computing/maintaining
that length instead of the actual cycle. The obtained algorithms, however, can be easily extended
to return a sought cycle with no additional asymptotic overhead.}
can be trivially reduced to APSP in $O(n^2)$
time by simply inspecting the entries of the distance matrix.
In fact, as shown by Vassilevska Williams and Williams~\cite{WilliamsW18}, for dense graphs APSP is known to be subcubically equivalent
to many problems which look easier at first sight, especially because their
output is just a single number (as opposed to $n^2$ numbers in APSP).
These include e.g., the radius, the minimum weight cycle, and the second shortest simple $s,t$ path problems.
For all these problems, just like for APSP, the best known algorithms run in $\Ot(nm)$ time.
Lincoln et al.~\cite{LincolnWW18} gave some compelling reasons why improving upon this
bound may also be impossible.

In this paper, our focus is on \emph{fully dynamic} graph algorithms. Fully dynamic
graph algorithms allow updating the graph under both edge insertions
and deletions, as opposed to \emph{partially dynamic} algorithms that
allow either only insertions (\emph{incremental} setting) or only deletions (\emph{decremental} setting).
Fully dynamic APSP has been widely studied in the past.
Demetrescu and Italiano~\cite{DemetrescuI04} showed that the distance matrix
can be \emph{explicitly} maintained in $\Ot(n^2)$ \emph{amortized} time under
\emph{vertex updates} which are allowed to change all edges incident
to a single vertex at once. Thorup~\cite{Thorup04b} simplified and slightly improved
their algorithm. 
Clearly, if the algorithm is required to maintain all distances explicitly,
one cannot break through the $O(n^2)$ time barrier since even a \emph{single} edge update may change \emph{all} the
$n^2$ pairwise distances.
Much of the work in this topic~\cite{AbrahamCK17, GutenbergW20b, Thorup05}
has been devoted to obtaining
good \emph{worst-case} bounds on the time needed to recompute the distance matrix
and it is known that $\Ot(n^{2+2/3})$ worst-case update time is possible~\cite{AbrahamCK17, GutenbergW20b}.
Interestingly, none of the known fully dynamic algorithms for \emph{real-weighted} dynamic
APSP has $o(n^2)$ update time and a non-trivial query procedure running in $o(m)$ time.
Such an algorithm, with $\Ot(m\sqrt{n})$ amortized update time and $\Ot(n^{3/4})$ query time,
has so far been only described for sparse enough unweighted graphs by Roditty and Zwick~\cite{RodittyZ11}.

The algorithm of Demetrescu and Italiano~\cite{DemetrescuI04} immediately
implies $\Ot(n^2)$ amortized update bound for fully dynamic variants of all
the most fundamental problems ``trivially reducible'' to APSP -- the aforementioned
diameter, radius, or minimum weight cycle.
Surprisingly, as shown in~\cite{AnconaHRWW19}, such an update
bound is likely to be the best possible for maintaining both the diameter and the radius (conditionally on so-called
Strong Exponential Time- and Hitting Set hypotheses~\cite{AbboudWW16, ImpagliazzoPZ01}), 
even if the graph remains sparse at all times and $(3/2-\eps)$-approximation
is allowed.

It is thus natural to ask whether there exist fully dynamic algorithms for
the \emph{minimum weight cycle problem} that improve upon the reduction to fully dynamic APSP for sparse graphs,
possibly allowing some small multiplicative approximation.
The fundamental difference between the minimum weight cycle and
diameter/radius problems is that the trivial reduction of minimum weight cycle requires
reading only $m$ entries of the distance matrix, as opposed to all $n^2$
in the case of radius and diameter.
As a result, by using the aforementioned fully dynamic algorithm of
Roditty and Zwick~\cite{RodittyZ11}, one immediately gets $\Ot(mn^{3/4})$
amortized update bound but merely for \emph{unweighted} graphs.
Note that this bound is always better than recompute-from-scratch,
and is truly subquadratic for sparse graphs.
It is however not clear whether such a bound can be obtained
for real-weighted graphs, nor whether a much better bound is attainable
if we allow approximation.

Motivated by the above, in this paper we initiate the study of the directed minimum weight cycle problem
in the fully dynamic setting.
To the best of our knowledge, this problem has not been explicitly studied in the literature before.
It is worth noting, however, that a non-trivial fully dynamic algorithm has been shown
for \emph{undirected planar} graphs~\cite{LackiS11}.

\newcommand{\mc}{\phi}
\newcommand{\thres}{\mu}

\subsection{Our results}

\paragraph{A fully dynamic approximate minimum weight cycle algorithm.}
Our first contribution is a simple deterministic fully dynamic algorithm
maintaining a $(1+\eps)$-approximation of the minimum weight $\mc(G)$ of
a cycle in a real-weighted directed graph~$G$.
If~$G$ has a negative cycle, then we define $\mc(G)=-\infty$, thus
allowing the sought cycle to be non-simple.
Note that if we wanted the minimum weight cycle to be simple and simultaneously
allowed negative edge weights, the problem would become NP-hard via
a reduction from Hamiltonian cycle.

\begin{restatable}{theorem}{apprcycle}\label{t:appr-cycle}
  Let $G$ be an initially empty fully dynamic real-weighted digraph such that 
  the weight of each positive weight cycle in $G$ always belongs to the interval $[c,C]$, $c,C\in\mathbb{R}$.

  There exists an algorithm maintaining an estimate $\mc'$ satisfying $\mc(G)\leq \mc'\leq (1+\eps)\mc(G)$
  under vertex updates to $G$ with amortized update time $O((m+n\log{n})\cdot \log{(C/c)}/\eps)$.
\end{restatable}

By~Theorem~\ref{t:appr-cycle}, a simpler amortized update time bound of $O((m+n\log{n})\cdot \log{(nW)}/\eps)$
for the fully dynamic $(1+\eps)$-approximate minimum weight cycle problem can be obtained in two special cases:
\begin{itemize}
  \item if $G$ has real-weights in $\{0\}\cup [1,W]$,
  \item if $G$ has integer weights in $(-\infty,W]$.
\end{itemize}

Via known conditional lower bounds on the static approximate
minimum weight cycle problem, the update time bound in Theorem~\ref{t:appr-cycle} -- as a function of $m$ \emph{alone} -- is near-optimal
for both vertex and edge updates
if we allow approximation factor less than $2$ and
$O(m^{2-\delta})$ preprocessing time (for some $\delta>0$).
Indeed, Dalirrooyfard and Vassilevska Williams~\cite{DalirrooyfardW20} proved
that under so-called $k$-Cycle hypothesis~\cite{AnconaHRWW19}, one cannot approximate the minimum weight
cycle within factor less than $2$ in $O(m^{2-\delta})$ time, for any $\delta>0$.
Clearly, if there was, say, a dynamic $3/2$-approximate minimum weight
cycle algorithm with $O(m^{2-\delta})$ preprocessing time, $O(m^{1-\delta})$ update time,
and the same interface as our algorithm,
$m$ edge/vertex updates would be sufficient to obtain
a static $3/2$-approximate minimum weight cycle algorithm running in $O(m^{2-\delta})$ time.
This would refute the $k$-Cycle hypothesis.

Observe that the $\Omega(m^{2-o(1)})$ conditional lower bound~\cite{DalirrooyfardW20}
(which implies that the $\Omega(mn^{1-o(1)})$ bound holds for \emph{some} sparsity $m$)
on the complexity of static approximate minimum weight cycle problem does not
rule out dynamic vertex update bounds of the form $\Ot(n^{\alpha}\cdot m^{1-\alpha})$
for some $\alpha\in (0,1]$ or $\Ot(m^{1+\beta}/n^{2\beta})$ for some $\beta\in (0,1/2]$, e.g., $\Ot(n)$, $\Ot(\sqrt{nm})$,
or $\Ot(m^{3/2}/n)$.
However, 
if we limit ourselves to
\emph{``combinatorial''} algorithms that
do not rely on fast matrix multiplication,
such $O(m^{1-\eps})$ bounds are ruled out for infinitely many sparsities of the form
$m=\Theta(n^{1+2/(k-1)})$, where $k\geq 3$ is an odd integer~\cite{DalirrooyfardW20, LincolnWW18}.

We stress that the aforementioned static conditional lower bounds do not rule out $\Ot(n)$ or even
$\Ot(\sqrt{nm})$ amortized
update time in the \emph{edge update} model.
In this case, for similar reasons, only combinatorial approximate algorithms with amortized update time
that is sublinear in $n$ for many sparsities, e.g., $\Ot(m^\beta\cdot n^{1-2\beta})$ for $\beta\in (0,1/2]$,
are unlikely to exist.

\paragraph{Fully dynamic cycles, flows, and cuts in planar graph.} Interestingly, if we limit our attention to the case of single edge updates (as opposed
to vertex updates) and real weights in $\{0\}\cup [1,W]$, the amortized update cost of the data structure of Theorem~\ref{t:appr-cycle} can always be charged to the cost of performing
a single edge update plus a single distance query on $O(\log{(nW)}/\eps)$ fully dynamic \emph{exact} distance oracles,
each maintaining some subgraph of~$G$.
For general digraphs, this amounts to running Dijkstra's algorithm in each
of these subgraphs since no non-trivial fully dynamic distance oracles with
both update and query time $o(m)$ are known.
However, such dynamic distance oracles are well-known to exist for planar digraphs~\cite{FR, KaplanMNS17, MSSP} which
immediately 
leads to the following result.

\begin{theorem}\label{t:planar-cycle-intro}
  Let $G$ be a planar digraph $G$ 
  with real weights in $\{0\}\cup [1,W]$.
  There exists an algorithm maintaining an $(1+\eps)$-approximate estimate of $\mc(G)$
  under planarity preserving edge insertions and deletions
  with amortized update time $\Ot(n^{2/3}\log{(W)}/\eps)$.
\end{theorem}

Previously, no sublinear fully dynamic algorithm for minimum weight cycle in planar
directed graphs has been described. An \emph{exact} algorithm for planar \emph{undirected} graphs
with $\Ot(n^{5/6})$ update time was given by Łącki and Sankowski~\cite{LackiS11}.

There is a well-known correspondence between simple cuts in an undirected plane graph~$G$,
and simple cycles in its dual $G^*$.
The correspondence, in a way, extends to \emph{directed} planar graphs (see e.g.~\cite{LiangL17, MozesNNW18}). Nevertheless,
currently the best known min $s,t$-cut algorithms in planar digraphs~\cite{BorradaileK09, Erickson10} are less efficient and
use entirely different techniques than their counterparts for planar undirected graphs~\cite{ItalianoNSW11}.
Generally speaking, for cut/flow applications, undirected planar graphs proved
much more friendly to work with (see e.g., the discussion in~\cite{Erickson10}~or~\cite{MozesNNW18}).
As an example of this phenomenon, an \emph{exact} fully dynamic max $s,t$-flow \emph{oracle} (accepting
$s,t$ as query parameters) with $\Ot(n^{2/3})$ update and query time
exists for undirected plane graphs~\cite{ItalianoNSW11},
whereas no such dynamic algorithm has been described for \emph{directed} plane graphs, even
allowing approximation and just a single fixed source-sink pair.

It is known that in a plane digraph $G$, an $s,t$-flow of value $f$ can be routed
if and only if the dual $G_{s,t,f}^*$ of a certain augmentation of $G_{s,t,f}$ depending on $s,t$ and $f$
contains no negative cycles~\cite{Erickson10, Johnson87, nussbaum2014network}.
Roughly speaking, since the algorithm of Theorem~\ref{t:planar-cycle-intro} supports negative weights,
by running it on
$G_{s,t,f}^*$ for $O(\log{(nW)}/\eps)$ distinct values of $f$, we obtain the following.

\begin{restatable}{theorem}{planarflow}\label{t:planar-flow-intro}
  Let $G$ be a plane embedded digraph 
  with real edge capacities in $\{0\}\cup [1,W]$ and a \emph{fixed} source/sink pair $s,t$.
  There exists an algorithm maintaining a $(1-\eps)$-approximate estimate of the value of maximum $s,t$-flow
  in $G$ under \emph{embedding preserving} edge insertions and deletions
  with $\Ot(n^{2/3}\log{(W)}/\eps)$ amortized update time.
\end{restatable}

To the best of our knowledge, the above constitutes the first known fully dynamic
maximum $s,t$-flow algorithm for plane directed graphs with a sublinear update time bound.

\paragraph{Exact fully dynamic minimum weight cycle and MPSP.}
Finally, we consider maintaining the minimum weight cycle \emph{exactly} in a fully dynamic \emph{real-weighted} digraph.
We show:

\begin{restatable}{theorem}{exactcycle}\label{t:exact-cycle}
  Let $G$ be a real-weighted digraph.
  There exists a Monte Carlo randomized fully dynamic algorithm maintaining
  $\mc(G)$ under vertex updates
  with $O((m+n\log{n})n^{2/3}\log^{4/3}{n})$ worst-case update time.
  The answers produced are correct with high probability.\footnote{That is, with probability at least $1-1/n^c$ for any
  chosen constant $c\geq 1$.}
\end{restatable}

Note that for sparse graphs, Theorem~\ref{t:exact-cycle} allows recomputing
the minimum weight cycle in $\Ot(n^{5/3})$ time, i.e., polynomially
faster than recompute-from-scratch and the dynamic algorithm of Demetrescu and Italiano~\cite{DemetrescuI04}.
However, observe that~\cite{DemetrescuI04} yields a better amortized update bound for $m=\omega(n^{4/3})$.

In order to obtain Theorem~\ref{t:exact-cycle}, we generalize the fully dynamic
APSP algorithm of Abraham et al.~\cite{AbrahamCK17} in a non-trivial way 
to solve what we call
the \emph{multiple pairs shortest paths} problem (\emph{MPSP}).
In the MPSP problem, which may be of independent interest, one requires to maintain only $k$ fixed
entries of the distance matrix, i.e., after each update we are interested in distances
between some source-target pairs $(s_i,t_i)$ for $i=1,\ldots,k$.
Recall that the minimum weight cycle of a directed graph can be computed by inspecting
distances for~$m$ source-target pairs.
We obtain the following bound for the fully dynamic MPSP problem.

\begin{restatable}{theorem}{tmpsp}\label{t:mpsp}
  Let $G$ be a real-weighted digraph.
  There exists a Monte Carlo randomized fully dynamic MPSP data structure
  supporting vertex updates
  with $O((m+n\log{n}+k)n^{2/3}\log^{4/3}{n})$ worst-case update time.
  The answers produced are correct with high probability.
\end{restatable}

Note that the aforementioned data structure of Roditty and Zwick~\cite{RodittyZ11} trivially
implies an MPSP data structure for \emph{unweighted} digraphs with $\Ot(m\sqrt{n}+kn^{3/4})$
amortized update bound.
Our result shows that a better (even worst-case) bound for (even real-weighted) sparse graphs can be achieved
if the set of source-target pairs is fixed throughout.

Actually, just as the worst-case update time of the data structure
of Abraham et al.~\cite{AbrahamCK17} can be very easily improved to $\Ot(n^{2.5})$
for \emph{unweighted} graphs~\cite[Section~4.2]{AbrahamCK17}, an unweighted variant of our MPSP data structure has
$\Ot((m+k)\sqrt{n})$ worst-case update time.

Interestingly, it seems that the other known approaches to fully dynamic APSP
in real-weighted graphs~\cite{DemetrescuI04, GutenbergW20, Thorup04b},
if adjusted,
cannot easily yield subquadratic (in $n$) update times for ``sparse'' instances
of MPSP where $m,k=O(n)$.
This is because they all reconstruct shortest paths in a hierarchical manner,
by inductively stitching~\cite{GutenbergW20} or extending~\cite{DemetrescuI04, Thorup04b}
paths recomputed earlier in the process.
Even though the number of input source-target pairs of interest may be small,
these may require answers for $\Theta(n^2)$ distinct source-target pairs at lower levels of the hierarchy.
The data structure of Abraham et al.~\cite{AbrahamCK17}, on the contrary,
does not use a hierarchical approach and can be thought as using a single ``stitching layer''.

Since the algorithm behind Theorem~\ref{t:exact-cycle} (Theorem~\ref{t:mpsp}) is exact, the
maintained information, i.e., the minimum weight of a cycle (the entries of the distance
matrix of interest, resp.) is unique.
Therefore, if we are interested in maintaining the corresponding weight (distances, resp.) only,
the bounds in Theorems~\ref{t:exact-cycle}~and~\ref{t:mpsp} hold against an adaptive adversary.
However, if we are required to output some 
actual minimum weight cycle (edges on some of the desired shortest paths, resp.)
we have to assume an oblivious adversary.\footnote{Abraham et al.~\cite{AbrahamCK17} show how to extend their data structure so that it
is capable of tracking
lexicographically smallest shortest paths and thus works against an adaptive
adversary, even when returning actual paths is required.
Out of the box, this additional feature costs $\Omega(n^2)$ extra time per update, though.
Adapting this idea to minimum weight cycle and MPSP is an interesting possible further step.
}

\subsection{Related work}
\paragraph{Computing minimum weight cycles statically.}
The best known algorithm for computing the minimum
weight cycle in sparse graphs exactly runs in $O(nm)$ time~\cite{OrlinS17}.
One can improve upon this for graphs with
small integer weights using matrix multiplication~\cite{CyganGS15, ItaiR78, RodittyW11}.
A subcubic-time $(1+\eps)$-approximation can also be achieved this way~\cite{BringmannKW19, Zwick02}.
Much of the recent work regarded approximating the minimum weight cycle within factor at least 2~\cite{ChechikL21, ChechikLRS20, DalirrooyfardW20}.

\paragraph{Dynamic APSP.} Apart from the fully dynamic setting, APSP has also been widely studied
in partially dynamic settings.
There exist efficient exact algorithms for \emph{unweighted} digraphs
with $\Ot(n^3)$ total update time
in both incremental~\cite{AusielloIMN91} and decremental~\cite{BaswanaHS07, abs-2010-00937} settings.
The fully dynamic APSP algorithm~\cite{DemetrescuI04, Thorup04b} is known
to have total update time $\Ot(n^3)$ in the decremental setting for real-weighted digraphs, but only when each update removes \emph{all} edges incident to a vertex (and thus there are at most $\leq n$ updates).
For weighted digraphs, a nearly optimal $\Ot(nm\log(W)/\eps)$ total update time partially dynamic algorithm is known
in the $(1+\eps)$-approximate setting~\cite{Bernstein16}.
This algorithm assumes an oblivious adversary though.
Less efficient algorithms that are either deterministic or assume an adaptive adversary are known~\cite{abs-2010-00937, KarczmarzL19, KarczmarzL20}.
Note that many of the above algorithms maintain the distance matrix explicitly
so they can be obviously used to maintain the minimum weight cycle (possibly approximately) in the respective partially dynamic scenarios.

Dynamic APSP has also been studied in undirected graphs~\cite{Bernstein09, BernsteinR11, Chechik18, ChuzhoyS21, GutenbergW20, Henzinger16, RodittyZ12}.

\subsection{Organization of the paper}
The rest of this paper is organized as follows.
In Section~\ref{s:prelims} we fix the notation. In Section~\ref{s:decision}
we show a fully dynamic \emph{threshold cycle detection} data structure that constitutes the
heart of the fully dynamic $(1+\eps)$-approximate minimum weight cycle algorithm
of Theorem~\ref{t:appr-cycle} proved in Section~\ref{s:appr}.
The applications of Theorem~\ref{t:appr-cycle} to planar graph algorithms,
in particular the proofs of Theorems~\ref{t:planar-cycle-intro}~and~\ref{t:planar-flow-intro},
are covered in detail in Section~\ref{s:planar}.
In Section~\ref{s:exact} we describe the exact fully dynamic minimum weight cycle
and fully dynamic MPSP algorithms.

\section{Preliminaries}\label{s:prelims}

In this paper we deal with \emph{real-weighted directed} graphs.
We write $V(G)$ and $E(G)$ to denote the sets of vertices and edges of $G$, respectively.
We denote by $n$ and $m$ numbers of vertices and edges (resp.) in the input graph.
A graph $H$ is a \emph{subgraph} of $G$, which~we~denote by $H\subseteq G$, if
and only if $V(H)\subseteq V(G)$ and $E(H)\subseteq E(G)$.
We write $uv\in E(G)$ when referring to edges of $G$ and use $\wei_G(uv)$
to denote the weight of $uv$.

For an edge set $F$, we sometimes write $G+F$ to denote the graph
$(V(G),E(G)\cup F)$.
If $F$ contains an edge $uv$ of weight $x$ and $uv\in E(G)$, then we assume
that $\wei_{G+F}(uv)=\min(\wei_G(uv),x)$.
For an edge $e$ we sometimes use $G+e$ to denote
$G+\{e\}$. 
For a subset $D\subseteq V$, we define $G\setminus D$ to be the graph
$G$ with all edges incident to vertices in $D$ removed.

A sequence of edges $P=e_1\ldots e_k$, where $k\geq 1$ and $e_i=u_iv_i\in E(G)$, is called
an $s\to t$ path in~$G$ if $s=u_1$, $v_k=t$ and $v_{i-1}=u_i$ for each $i=2,\ldots,k$.
For brevity we sometimes also express $P$ as a sequence of $k+1$ vertices $u_1u_2\ldots u_kv_k$ or as a subgraph of $G$ with vertices $\{u_1,\ldots,u_k,v_k\}$
and edges $\{e_1,\ldots,e_k\}$.
A path $P$ is \emph{simple} if $u_i\neq u_j$ for $i\neq j$.
A \emph{cycle} is a path such that $u_1=v_k$.
A~\emph{simple cycle} is a cycle that is a simple path.

The \emph{hop-length} of $P$ is the number of edges in $P$.
We also say that $P$ is a \emph{$k$-hop path}.
The \emph{length} of the path $\len(P)$ is defined as $\len(P)=\sum_{i=1}^k\wei_G(e_i)$.
For convenience, we sometimes consider a single edge $uv$ as a path of hop-length $1$.
If $P_1$ is a $u \to v$ path and $P_2$ is a $v \to w$ path, we denote by $P_1\cdot P_2$ (or simply $P_1P_2$) a path obtained by concatenating $P_1$ with $P_2$.

The \emph{distance} $\dist_G(u,v)$ between the vertices $u,v\in V(G)$ is the length
of the shortest $u\to v$ path in $G$, or $\infty$, if no $u\to v$ path
exists in $G$.

Note that the distance is well-defined only if~$G$ contains no negative cycles.
It is well-known that (1) $G$ has no negative cycles if and only if
there exists a \emph{feasible price function} $p:V\to\mathbb{R}$ satisfying
$\wei_G(e)+p(u)-p(v)\ge 0 $ for all $uv=e\in E(G)$,
and (2) given a feasible price function of $G$, one can compute single-source
shortest paths in $G$ using Dijkstra's algorithm even if $G$ has edges with negative weights.

Define $\mc(G)$ to be the infimum of $\len(C)$ through all cycles $C\subseteq G$.
Note that here $C$ is not necessarily a simple cycle: in general finding minimum weight simple cycles with arbitrary negative weights is NP-hard.
In particular, if $G$ contains no cycles at all, then we define $\mc(G):=\infty$.
If $G$ contains a negative cycle, then $\mc(G)=-\infty$.
On the other hand, if $\mc(G)\geq 0$, then $G$ contains a simple cycle $C'$ with $\len(C')=\mc(G)$.
We call any such cycle~$C'$
a \emph{minimum weight cycle}.
Observe that if $\mc(G)\geq 0$, then $\mc(G)=\min_{uv\in E(G)}\{\dist_G(v,u)+\wei_G(uv)\}$.
\begin{observation}\label{obs:through}
  Let $H$ be a non-negatively weighted digraph and let $v$ be its vertex.
  The minimum weight of a cycle in $H$ that goes through $v$ can be
  computed in $O(m+n\log{n})$ time.
\end{observation}
\begin{proof}
  First compute single-source shortest paths from $v$ using Dijkstra's algorithm.
  Note that the minimum weight cycle through $v$ has length
  $\min_{uv\in E(H)}\{\dist_H(v,u)+\wei_H(uv)$\}.
\end{proof}

When characterizing dynamic graph algorithms, we use the term \emph{edge update}
to refer to a graph update that changes (i.e., inserts, removes, or alters the weight) a single edge of $G$.
On the other hand, a \emph{vertex update} can change all edges
incident (incoming or outgoing) to a single chosen vertex $v\in V(G)$.
In this case, we say that such a vertex update is \emph{centered}~at~$v$.

\section{Fully dynamic threshold cycle detection}\label{s:decision}
Consider the following decision variant of the fully dynamic minimum
weight cycle problem.
Suppose we would like to maintain the information whether the minimum weight
$\mc(G)$ of a cycle in a real-weighted digraph $G$ is below some threshold $\thres\geq 0$.
In this section we show:
\begin{theorem}\label{t:threshold}
  Let $G$ be an initially empty real-weighted digraph and let $\mu\geq 0$. There exist a fully dynamic algorithm
    maintaining the information whether $\mc(G)<\mu$ and supporting vertex updates in $O(m+n\log{n})$ amortized time.
\end{theorem}

The idea is to keep the edge set $E$ partitioned into two subsets $E_0$ and $E_1$ such that
the following two invariants are satisfied:
\begin{enumerate}[label=(\arabic*)]
  \item For $G_0=(V,E_0)$ we have $\mc(G_0)\geq \thres$.
  \item If $E_1\neq\emptyset$, then $\mc(G)<\thres$.
\end{enumerate}
Observe that by the above invariants, $\mc(G)<\thres$ \emph{if and only if} $E_1\neq\emptyset$.

Let us first consider the case when $G$ has non-negative edges only.
Then we can assume $\mu>0$ since the answer for $\mu=0$ is trivially ``no''.

We store $E_1$ partitioned into subsets $E_1(v)$ for $v\in V$,
so that each edge $uv\in E_1$ is stored in either $E_1(u)$ or $E_1(v)$
(this choice is arbitrary).
Since the data structure is initialized with an empty graph, initially $E_0=\emptyset$ and $E_1(v)=\emptyset$ for all $v\in V$.

We also store the vertices $v$ with $E_1(v)\neq \emptyset$ of $G$ in a list $Q$ sorted by the
time when the last insertion around $v$ happened, i.e., 
at the end of $Q$ we have a vertex that has been most recently subject
to insertion of edges around $v$.

Let us now describe an auxiliary procedure $\texttt{update}(v)$ that will be used to fix the invariants.
$\texttt{update}(v)$ does the following.
We assume that $E_1(v)\neq \emptyset$.
We compute the minimum weight~$x$ of a cycle going through~$v$ in $G_0+E_1(v)=(V,E_0\cup E_1(v))$
as described in Observation~\ref{obs:through}.
If $x\geq \thres$, the edges $E_1(v)$ are moved to $E_0$,
and the set $E_1(v)$ is emptied.
This change is reflected in $Q$ by removing $v$ from $Q$.

To handle and insertion of a set $F_v$ of edges centered at some vertex $v$,
we simply add the edges $F_v$ to $E_1(v)$, move $v$ to the end of $Q$, and, if $Q=\{v\}$, run $\texttt{update}(v)$.

To handle a deletion of an arbitrary set
of edges $F\subseteq E$, we first remove each edge $f\in F$
from $E_0$ or some set $E_1(w)$, wherever $f$ resides.
If some $E_1(w)$ is emptied this way, $w$ is removed from $Q$ accordingly.
Next, while $Q\neq\emptyset$, we repeatedly run $\texttt{update}(v)$ for the first element $v\in Q$
and stop if $\texttt{update}(v)$ fails to empty the respective set $E_1(v)$.

We now prove the correctness of the algorithm, whose pseudocode is given in Algorithm~\ref{alg:alg}.
\begin{algorithm}[h!]\caption{Detecting a cycle of weight less than $\thres$.}\label{alg:alg}

\label{alg1}
  \textbf{procedure} $\texttt{update}(v)$
  \begin{algorithmic}[1]
  \STATE $x:=\text{the minimum weight of a cycle going through } v \text{ in } G_0+E_1(v)$
  \IF{$x\geq\thres$}
    \STATE $E_0:=E_0\cup E_1(v)$
    \STATE $E_1(v):=\emptyset$
    \STATE $Q:=Q\setminus\{v\}$
  \ENDIF
\end{algorithmic}
  \vspace{2mm}
\textbf{procedure} $\texttt{insert}(F_v\neq\emptyset)$

  \begin{algorithmic}[1]
  \STATE $E_1(v):=E_1(v)\cup F_v$
  \STATE $\texttt{move-to-back}(Q,v)$
  \IF{$Q=\{v\}$}
    \STATE $\texttt{update}(v)$
  \ENDIF
\end{algorithmic}

  \vspace{2mm}
  \textbf{procedure} $\texttt{delete}(F\subseteq E(G))$

  \begin{algorithmic}[1]
  \STATE $E_0:=E_0\setminus F$
  \FOR{$uv=e\in F$}
    \FOR{$w\in \{u,v\}$}
      \STATE $E_1(w):=E_1(w)\setminus \{e\}$
      \IF{$E_1(w)=\emptyset$}
        \STATE $Q:=Q\setminus\{w\}$
      \ENDIF
    \ENDFOR
  \ENDFOR
  \WHILE{$Q\neq\emptyset$}
    \STATE $v:=\texttt{front}(Q)$
    \STATE $\texttt{update}(v)$
    \IF{$E_1(v)\neq\emptyset$}
      \STATE \textbf{break}
    \ENDIF
  \ENDWHILE
\end{algorithmic}
  \vspace{2mm}
\textbf{function} $\texttt{cycle-below-threshold}()$
  \begin{algorithmic}[1]
  \RETURN $Q\neq \emptyset$
\end{algorithmic}
\end{algorithm}

\begin{observation}\label{o:go-through}
  Suppose $\mc(G_0)\geq \thres$ and let $v\in V$. Then $\mc(G_0+ E_1(v))<\thres$
  if and only if the shortest cycle going through $v$ in $G_0+E_1(v)$
  has weight less than $\thres$.
\end{observation}
\begin{proof}
  By $\mc(G_0)\geq\thres$, a cycle of weight less than $\thres$ in $G_0+ E_1(v)$
  has to go through an edge of $E_1(v)$.
  All of these edges are incident to the vertex $v$.
\end{proof}

Clearly, $E_0$ and $E_1$ form a partition of $E$ after each insertion or deletion:
the procedure~\texttt{update} only moves edges from $E_1$ to $E_0$.
\begin{lemma}\label{l:inv1}
  Invariant~(1) is maintained throughout the updates.
\end{lemma}
\begin{proof}
Note that no edge is added to $E_0$ outside the $\texttt{update}$
procedure. As a result, since invariant~(1) cannot be broken
by removing edges from $G_0$, to establish that invariant~(1) is maintained,
it is enough to see that
$\texttt{update}$ only adds edges to $E_0$ if $\mc(G_0)\geq\thres$ afterwards.
\end{proof}
\begin{lemma}\label{l:inv2}
  Invariant~(2) is maintained throughout the updates.
\end{lemma}
\begin{proof}
Let $G',G_0',E_1'$ denote $G,G_0,E_1$ respectively \emph{before} the graph update.
Suppose that after processing the update, invariant~(2) is broken. Equivalently, $E_1\neq \emptyset$ and $\mc(G)\geq \thres$.

Suppose the update was insertion of edges $F_v$ centered at $v$.
  Since adding edges can only decrease the minimum weight of a cycle,
$\mc(G')\geq\thres$. As invariant~(2) was satisfied before,
$E_1'=\emptyset$.
  So after $F_v$ is moved to $E_1(v)$, we indeed have $Q=\{v\}$.
Since $\mc(G_0+E_1(v))\geq \mc(G)\geq \thres$, $E_1(v)$ should have been moved to $E_0$ by $\texttt{update}(v)$.
But $E_1=E_1\setminus E_1'\subseteq E_1(v)=\emptyset$, so $E_1=\emptyset$, a contradiction.

Now assume that the update deleted an arbitrary subset of edges.
  If after some $\texttt{update}(v)$ call we have $E_1(v)\neq\emptyset$,
  then $\mc(G_0+E_1(v))<\thres$, which implies $\mc(G)<\thres$,
  a contradiction.
  If no such $v$ exists, then $Q$ is emptied, i.e., $E_1(v)=\emptyset$ for all $v\in V$
  after the deletion is processed.
  It follows that $E_1=\emptyset$, which again leads to a contradiction.
\end{proof}

Now let us analyze the running time of our algorithm.
\begin{lemma}\label{l:insert}
  Each insertion is processed in $O(m+n\log{n})$ worst-case time.
\end{lemma}
\begin{proof}
  An insertion adds $O(n)$ edges to a single set $E_1(v)$ and causes 
  at most a single $\texttt{update}$ call. The running time
  of $\texttt{update}$ is dominated by the time needed to find
  the minimum weight of a cycle going through some vertex $v$
  in some subgraph of the current graph $G$.
  By Observation~\ref{obs:through}, this time is no more than $O(m+n\log{n})$.
\end{proof}

\begin{lemma}\label{l:threshold-amortized}
  The total time needed to process arbitrary $k$ updates is
    $O\left(\sum_{i=1}^k(m_i+n\log{n})\right)$,
  where $m_i$ is the number of edges in $G$ when the $i$-th update happened.
  In other words, the amortized update time is $O(m+n\log{n})$.
\end{lemma}
\begin{proof}
  By Lemma~\ref{l:insert}, we only need to prove that the deletions
  take $O\left(\sum_{i=1}^k(m_i+n\log{n})\right)$ time in total.
  The cost of removing the edges from the sets $E_0$ and $E_1(w)$, $w\in V$,
  can be charged to the insertions which added those edges
  to the graph.

  After updating the edge set,
  a deletion is handled using a number of $\texttt{update}(v)$ runs,
  in the order in which vertices $v$ appear in $Q$.
  At most one of these runs leaves $E_1(v)$ non-empty afterwards.
  We charge the cost of this run to the considered deletion.
  For all other $\texttt{update}(v)$ runs during that deletion,
  they empty the set $E_1(v)$ that previously was non-empty.
  As a result, we can charge the cost of that run to
  the last insertion of edges centered at $v$ that happened
  before the considered deletion.

  We need to prove two things.
  First of all, to see that no insertion is charged twice, note that
  after an insertion is charged for the first time, $E_1(v)$ is emptied.
  So, before $\texttt{update}(v)$ is called next time
  when handling a deletion, new edges have to be added to $E_1(v)$,
  which can only happen during another later insertion centered at $v$.

  We also have to prove that just before $E_1(v)$ is emptied in $\texttt{update}(v)$, the
  number of edges in $G_0+ E_1(v)$ is $O(|E'|)$, where $E'$
  is the edge set of $G$ immediately after the last insertion $I$ centered at $v$ happened.
  To this end, we prove $E(G_0)\cup E_1(v)\subseteq E'$.
  
  We clearly had $E_1(v)\subseteq E'$ immediately after $I$.
  Afterwards no more elements were added  to $E_1(v)$ (albeit some might have been removed),
  so we still have $E_1(v)\subseteq E'$.
  
  Now suppose there is an edge $e\in E(G_0)$ with $e\notin E'$.
  Then, since $G_0\subseteq G$, $e$ was inserted into $G$ after the insertion $I$,
  as a result of a later insertion $I'$ centered at some $w\neq v$.
  The edge $e$ could have been added to $G_0$ only if
  $E_1(w)$ was emptied inside $\texttt{update}(w)$ immediately afterwards,
  but before $\texttt{update}(v)$ was called.
  Since $I$ was the last insertion centered at $v$ before $\texttt{update}(v)$
  was called, both $v$ and $w$ were in $Q$ when $\texttt{update}(w)$ was called.
  This is a contradiction: $\texttt{update}$
  is always called on the earliest element of $Q$, whereas the fact that
  $I$ happened before $I'$ implies that $v$ lied earlier than $w$ in $Q$
  when $\texttt{update}(w)$ was called.
\end{proof}

\begin{remark}
  When handling a deletion, we could in principle call $\texttt{update}(v)$
  for vertices~$v$ with $E_1(v)\neq\emptyset$ in arbitrary order,
  as opposed to in the order of least recent centered insertions.
  However, then one could only show a weaker total update time bound
  of $O(k(m_{\max}+n\log{n}))$, where $m_{\max}$ is the maximum number
  of edges in $G$ during the first $k$ updates.
\end{remark}

\subsection{Negative weights}\label{s:negative}
In this section we extend the obtained basic algorithm to also work
with negative edges. Recall that we still assume $\thres\geq 0$.
Note that the case $\thres=0$ is equivalent to dynamically maintaining whether
$G$ has a negative cycle.
Recall that if $G$ has a negative cycle, $\mc(G)=-\infty$.

Unfortunately, in presence of negative weights or cycles
we cannot simply use the algorithm behind Observation~\ref{obs:through}
to find the minimum weight cycle through a vertex~$v$
in $G_0+E_1(v)$ as we did in $\texttt{update}(v)$.
Instead, we use the following lemma.
\begin{lemma}\label{l:price}
  Let $H$ be a digraph with no negative cycles. Let $p:V\to\mathbb{R}$ be a feasible
  price function of $H$.
  Let $F$ be a set of edges centered at some vertex $v$.

  Then in $O(m+n\log{n})$ time one can find the minimum weight of a cycle going through~$v$ in $H+F$.
  Moreover, if $H+F$ contains no negative cycles, within
  the same time bound one can produce a feasible price function on $H+F$.
\end{lemma}
\begin{proof}
  Clearly, since $H$ has no negative cycles, a negative cycle in $H+F$
  has to go through~$v$.
  Let $E_v^+$ be the set of edges in $H+F$ incoming to $v$.
  Note that $H'=H+F-E_v^+$ has no negative cycles.
  Moreover, since $H'$ differs from $H$ by edges incident to $v$,
  the edge costs reduced by $p$ are non-negative for all edges
  of $H'$ possibly except the outgoing edges of~$v$.
  However, since $v$ has no incoming edges in $H'$,
  a price function $p'$ obtained from $p$ by sufficiently increasing $p(v)$
  (e.g., to $\max\{p(u)-\wei(vu):vu\in E(H')\}$)
  is a feasible price function of $H'$.
  With price function $p'$ in hand, we can compute distances from
  $v$ in $H'$ using Dijkstra's algorithm in $O(m+n\log{n})$ time.
  
  Now let $x=\min_{uv\in E_v^+}\{\dist_{H'}(v,u)+\wei_{H'}(uv)\}$.
  Observe that $x$ is indeed the minimum weight of a simple cycle
  in $H+F$.
  Moreover, $x\geq 0$ implies that $p^*(y):=\dist_{H'}(v,y)=\dist_{H+F}(v,y)$ is
  a feasible price function on the induced subgraph of $(H+F)[R]$ reachable
  from $v$.
  To extend that price function $p^*$ on $R\subseteq V$ to entire $V$,
  it is enough to set $p^*(z)=p(z)+M$ for all $z\in V\setminus R$, where
  $M$ is a sufficiently large number.
  To see that, note that $p^*$ is clearly a feasible
  price function on $(H+F)[R]$, $(H+F)[V\setminus R]$, and
  there are no edges from $R$ to $V\setminus R$ in $H+F$.
  For edges $zy\in E(H+F)\cap ((V\setminus R)\times R)$ we have
  $\wei_{H+F}(zy)+p^*(z)-p^*(y)=\wei_{H+F}(zy)+p(z)+M-p^*(y)$.
  For
  \begin{equation*}
    M=\max\{p^*(y)-p(z)-\wei_{H+F}(zy):zy\in E(H+F)\cap ((V\setminus R)\times R)\},
  \end{equation*}
  all the required reduced costs are non-negative.
\end{proof}

Now, given Lemma~\ref{l:price}, we modify the basic algorithm as follows.
In addition to the partition of $E$ into $E_0$ and $E_1$, we always
maintain a feasible
price function $p_0$ on $G_0$.
Then, in $\texttt{update}(v)$, we use Lemma~\ref{l:price}
to find the minimum weight $x$ of a cycle in $G_0+ E_1(v)$.
If the edges $E_1(v)$ are moved to $E_0$ (and thus $x\geq 0$ since $\mc(G_0+E_1(v))\geq \thres\geq 0$),
we update the price function~$p_0$ to that produced
by Lemma~\ref{l:price}.
Since the worst-case cost of running the algorithm from
Lemma~\ref{l:price} matches that of Observation~\ref{obs:through},
the time analysis remains unchanged.
Lemmas~\ref{l:inv1},~\ref{l:inv2},~\ref{l:threshold-amortized}~and~\ref{l:price}
together imply Theorem~\ref{t:threshold}.

\begin{remark}
For the problem of fully dynamically maintaining
the information whether~$G$ contains a \emph{negative} cycle (i.e., the special case $\thres=0$) there exists a better
  algorithm with $O(m+n\log{n})$ \emph{worst-case} (as opposed to only amortized) update time bound
   (see Theorem~\ref{t:worst-case-cycle}).
In fact, we make use of that algorithm when obtaining exact algorithms
with good worst-case bounds in Section~\ref{s:exact}.
The main idea is to generalize the problem to maintaining a minimum cost
circulation in the graph $G$ with imposed unit vertex/edge capacities (the details
  can be found in Section~\ref{s:neg-worstcase}). This resembles Gabow's reduction of single-source shortest paths with
  negative weights to the minimum cost perfect matching problem~\cite{Gabow85}.
However, the min-cost circulation based algorithm is not as robust when it
comes to obtaining fully dynamic algorithms for planar graphs (described in Section~\ref{s:planar}).
\end{remark}

\section{A fully dynamic $(1+\eps)$-approximate algorithm}\label{s:appr}
In this section we show how Lemma~\ref{l:threshold-amortized} can be used to
obtain an $(1+\eps)$-approximate minimum weight cycle algorithm,
for any $\eps\in (0,1]$.
Suppose $c\in\mathbb{R}$ ($C\in\mathbb{R}$) is a lower bound (an upper bound, respectively) on the weight of a \emph{positive}
cycle in $G$.

Suppose first that $G$ has positively weighted edges.
In order to convert the decision version from Section~\ref{s:decision},
all we have to do is to run it simultaneously
with $\thres=(1+\eps)^k$ for all integers $k=\lceil \log_{1+\eps}(c)\rceil,\ldots,\lceil \log_{1+\eps}(C)\rceil$.
To maintain an approximate minimum weight of a cycle $G$,
one only needs to keep track of the minimum $k$ such that
the fully dynamic decision algorithm for $(1+\eps)^k$ returns
yes. If no such $k$ exists, $G$ is acyclic since $\mc(G)<\infty$
implies $\mc(G)\leq C$. Otherwise, we have $(1+\eps)^{k-1}\leq \mc(G)<(1+\eps)^k$,
so indeed $(1+\eps)^k$ approximates $\mc(G)$ with multiplicative error
no more than $(1+\eps)$.
Since each of the $O(\log_{1+\eps}(C)-\log_{1+\eps}(c))=O(\log{(C/c)}/\eps)$ decision algorithms has $O(m+n\log{n})$ amortized update time,
the amortized time of the approximate algorithm
is $O((m+n\log{n})\log{(C/c)}/\eps)$.

The same bound can be achieved even if $G$ has non-positive edges
(without, however, changing the definition of $c$ and $C$)
by extending each threshold data structure as described in
Section~\ref{s:negative}.
Apart from the data structures for thresholds $\thres=(1+\eps)^k$,
we also need two more threshold cycle detection data structures:
one for $\thres=0$ to detect a negative cycle, and one
for $\thres=c$ to detect whether $\mc(G)=0$.
We have thus proved Theorem~\ref{t:appr-cycle}. 

\apprcycle*

\section{Dynamic algorithms for cycles, cuts and flows in planar graphs}\label{s:planar}
In this section we argue that the fully dynamic threshold cycle detection
algorithm can be implemented on planar directed graphs using the known
dynamic distance oracles on planar graphs.
Since the reduction in Section~\ref{s:appr} uses the threshold
data structure in a black-box way, this will imply
an $(1+\eps)$-approximate minimum weight cycle algorithm.

Using known reductions based on plane duality,
this will yield  fully dynamic $(1+\eps)$-approximate algorithms
for maintaining (1) the capacity of a global min-cut in a plane digraph,
(2) the value of maximum $s,t$-flow in a plane digraph.

The algorithms in this section handle \emph{edge updates}, as opposed
to more general vertex updates as was the case in the previous sections.
Observe that achieving sublinear update time for vertex updates is
not possible in general since a vertex update may need up to $\Theta(n)$
space to be described. 
More concretely, we will allow a single update to either insert or remove a
single edge $uv$, provided that this update preserves planarity
of $G$. In the cut/flow applications we will additionally
need to assume that the edge insertions are \emph{embedding preserving},
i.e., $u$ and $v$ lie on a single face of the current embedding of $G$.

Kaplan et al.~\cite{KaplanMNS17}, based on earlier work~\cite{FR, MSSP}, showed a dynamic distance oracle for
real-weighted plane graphs undergoing edge weight updates.
As argued in~\cite{Charalampopoulos20}, their bound also holds if arbitrary, not necessarily
embedding-preserving, edge updates are allowed.

\begin{theorem}[see~\cite{Charalampopoulos20, FR, KaplanMNS17, MSSP}]\label{t:dist-oracle}
Let $G$ be a real-weighted planar digraph.
There exists a fully dynamic algorithm supporting edge insertions and deletions
in $\Ot(n^{2/3})$ worst-case time, such that for any query vertices $s,t$,
the shortest $s\to t$ path in $G$ can be computed in $\Ot(n^{2/3})$ time.
If an edge insertion creates a negative cycle in $G$, the update algorithm
reports it and refuses to perform that insertion.
Edge insertions are not required to be embedding preserving.
\end{theorem}

\paragraph{Fully dynamic threshold- and minimum weight cycles.}
Consider using the fully dynamic threshold cycle detection
algorithm of Section~\ref{s:decision} in the edge update scenario.
Suppose that that algorithm attempts to moves edges from $E_1$ to $E_0$
single edge at a time.
This does not influence correctness; the efficiency of processing a \emph{node update}
could deteriorate though (which we do not mind).
Then, the \emph{amortized} update time to process the update involving 
an edge $uv$ can be actually bounded by the sum of times needed to:
\begin{enumerate}
  \item update the set $E_1(u)$ to reflect the graph update,
  \item if $uv$ is deleted, remove $uv$ from $G_0$,
  \item for some $xy\in E$, find the minimum weight of a cycle
    going through $xy$ in $G_0+xy$,
  \item if $\mc(G_0+xy)\geq\thres$, insert the edge $xy$ into $G_0$.
\end{enumerate}
Clearly, item~1 takes constant time.
If we store the (planar) graph $G_0$ in the data structure of
Theorem~\ref{t:dist-oracle}, items 2-4 above all require $\Ot(n^{2/3})$ time.
Indeed, items 2~and~4 translated to a single edge update to that data structure, whereas
item 3 amounts to computing $\dist_{G_0}(y,x)+\wei_G(xy)$
using a single query.
We thus obtain the following analogue of Theorem~\ref{t:threshold}.
\begin{theorem}\label{t:threshold-planar}
  Let $G$ be a real-weighted planar digraph and let $\mu\geq 0$. There exist a fully dynamic algorithm
    maintaining whether $\mc(G)<\mu$ and supporting planarity-preserving edge insertions and deletions in $\Ot(n^{2/3})$ amortized time.
\end{theorem}
Since Theorem~\ref{t:appr-cycle} uses the threshold data structure in a black-box way, we obtain:
\begin{theorem}\label{t:appr-cycle-planar}
  Let $G$ be a fully dynamic real-weighted planar digraph $G$ such that 
  the weight of any positive cycle in $G$ always lies in the interval $[c,C]$.
  
  There exists an algorithm maintaining the minimum weight cycle
  in $G$ under planarity preserving edge insertions and deletions
  with amortized update time $\Ot(n^{2/3}\log{(C/c)}/\eps)$.
\end{theorem}
Note that Theorem~\ref{t:appr-cycle-planar} immediately implies Theorem~\ref{t:planar-cycle-intro}.

\paragraph{Fully dynamic directed cuts and flows.}
Let $G$ be a \emph{plane embedded} digraph with real edge capacities in $\{0\}\cup [1,W]$.
Assume that every edge $e$ in $G$ has its reverse $e^R$ of capacity~$0$ embedded
into the same curve. We can then think of any edge as traversable in both
directions, but the cost of such a traversal is $0$ if the edge is traversed
in the reverse direction.
This assumption clearly does not influence values of
max-flows or min-cuts in~$G$, but makes the dual graph~$G^*$ possess certain useful properties.
We call a cycle in $G^*$ \emph{non-trivial} if it is not of the form
$ee^R$ for some edge $e\in E(G^*)$ and its reverse $e^R$.

We now state well-known properties relating flows/cuts in $G$ to cycles in the dual $G^*$.
\begin{lemma}[see e.g.~\cite{LiangL17}]\label{l:mincut}
  The global minimum cut in a plane graph $G$ corresponds to the minimum weight non-trivial
  cycle in $G^*$.
\end{lemma}

\begin{lemma}[see~\cite{Erickson10, Johnson87, nussbaum2014network}]\label{l:flow} Let $G$ be a plane digraph with some fixed source $s$ and sink $t$.
  For $f\geq 0$, let $G_{P,f}$ be a plane graph obtained from $G$ adding an embedded $s\to t$ path $P$
  such that for each edge $e$ of~$P$, the capacity of $e$ is $f$, whereas
  the capacity of $e^R$ is $-f$.
  
  There exists an $s,t$-flow of value $f$ in $G$ if and only if
  the dual $G_{P,f}^*$ of $G_{P,f}$ does not contain negative cycles.
\end{lemma}

By Lemma~\ref{l:mincut}, maintaining the (approximate) global min-cut dynamically under edge
\emph{embedding preserving} insertions/deletions
can be reduced to maintaining the (approximate) minimum weight non-trivial cycle
in the dual under vertex splits and edge contractions.

Let us now explain how such operations can be simulated using $O(1)$ updates
to the data structure of Theorem~\ref{t:appr-cycle-planar}
maintained on a certain augmented version $G^*_1$ of $G^*$,
so that the minimum weights of a non-trivial cycle in $G^*$ and $G^*_1$ are equal.
A similar reduction has been previously described in~\cite{ItalianoNSW11,LiangL17}.
Each vertex $v$ of the dual $G^*$ corresponds in $G^*_1$ to a path~$P_v$ of $\deg_{G^*}(v)$ 
vertices connected using $0$-weight edges
traversable in both directions.
For an edge $vu\in E(G^*)$ that is the $i$-th in (some) clockwise edge ring
of $v$, and $j$-th in (some) clockwise edge ring of $u$,
the $i$-th vertex of $P_v$ is connected by an edge of weight $\wei_{G^*}(vu)$
with the $j$-th vertex of $P_u$.
This way, (1) each vertex of $G_1^*$ has constant degree,
(2) each non-trivial cycle in $G^*$ has a corresponding non-trivial cycle of
the same weight in~$G^*_1$, (3) 
no additional (with respect to $G^*$) non-trivial cycles are introduced in $G^*_1$.

It is not hard to verify that each edge contraction or vertex split
in $G^*$ can be reflected using $O(1)$ edge insertions or deletions
issued to $G_1^*$.

Observe that the additional constraint that the minimum weight cycle
is non-trivial does not introduce any serious difficulties: in the data structure
of Theorem~\ref{t:threshold-planar} we compute the minimum weight
cycle through some edge $e$ by issuing a distance query to a graph
that \emph{does not contain} that edge. However, since
a minimum weight non-trivial cycle through $e$ in $G_0+e$ can traverse
$e$ in any of the two directions, we need to issue
two distance queries instead of one.
Similarly, if the minimum weight of a non-trivial cycle in $G_0+e$ is
at least $\mu$, we add both edges $e$ and $e^R$ (with appropriate weights)
to the distance oracle maintaining $G_0$.

We thus obtain the following theorem.
\begin{theorem}\label{t:min-cut}
  Let $G$ be a plane digraph with real capacities in $\{0\}\cup [1,W]$.
  There is an algorithm maintaining a $(1+\eps)$-approximate estimate of the capacity of the global min-cut of~$G$ 
  under \emph{embedding preserving} edge updates
  with $\Ot(n^{2/3}\log{W}/\eps)$ amortized update time.
\end{theorem}

To obtain a dynamic max $s,t$-flow algorithm, we use Lemma~\ref{l:flow}.
We keep track of whether there exists a negative
cycle (i.e., we set $\thres=0$) in the dual of a graph a $G_{P,f}$, where $f=(1+\eps)^k$,  for each
$k=0,\ldots,\lceil \log_{1+\eps}{(nW)}\rceil$.
Similarly as was the case for global min-cut, one can simulate
the effect that an embedding preserving edge update in $G$ has
on the negative cycles of the dual of $G_{P,f}$ using $O(1)$ updates
to the data structure of Theorem~\ref{t:threshold-planar}
maintained 
on an analogous augmentation $(G_{P,f})^*_1$ of $G_{P,f}^*$.

There is one subtle detail about how $G_{P,f}$ is updated
when $G$ is subject to embedding preserving edge insertions and deletions.
Note that Lemma~\ref{l:flow} requires us to embed any additional
simple $s\to t$ path $P$ into $G$. Embedding $P$ into $G$ subdivides
some of the original faces of~$G$. As a result, an edge $uv$ to be inserted
inside some face $F$ of $G$ may cross some edges of the currently used path~$P$ in $G_{P,f}$.
We deal with this problem as follows.
We~maintain an additional invariant that (the embedding of) the simple path $P$ crosses each face of the plane graph $G$ at most once.

Now, when a new edge $uv$ is inserted inside $F$, and $P$ has an edge $e=xy$ inside
$F$ that would cross $uv$, we first remove $e$ from $G_{P,f}$ to allow the insertion of $uv$.
This insertion splits~$F$ into two faces $F_1,F_2$ such that 
$x$ lies on $F_1$ and $y$ lies on $F_2$.
We now reconnect the path~$P$ by embedding two edges $xu$, $uy$ with appropriate
capacities as required by Lemma~\ref{l:flow}.

On the other hand, when an edge $uv$ is removed, two faces $F_1$ of $F_2$ of $G$
are merged into a single face $F$.
If at most one of them $F_1,F_2$ contained an edge of $P$, we do not have to do anything.
Otherwise, suppose wlog. that $F_1$ contains an edge $xy=e_1\in P$, and $F_2$
contains an edge $ab=e_2\in P$, such that $e_1$ appears before $e_2$ on $P$.
Then, we remove $e_1$, $e_2$, and all edges between $e_1$ and $e_2$ on $P$ from $G_{P,f}$,
and replace them with a single edge $xb$ embedded in $F$.
Afterwards, the invariant is satisfied and $P$ remains a simple path.

Finally, observe that each update to $G$ 
adds $O(1)$ new edges to $P$ in the \emph{worst case}.
An edge deletion may remove a superconstant number of edges from $P$,
but these removals can be charged to the corresponding
additions of new edges to $P$.
To conclude, an edge update to $G$ translates to $O(1)$ amortized
edge updates to $G_{P,f}$, and as a result, to $O(1)$ amortized
operations on the data structure of Theorem~\ref{t:threshold-planar}
run on the augmented dual $(G_{P,f})^*_1$.
We have thus proved:
\planarflow*

\section{Exact fully dynamic algorithm for minimum weight cycle}\label{s:exact}
In this section we argue that using a variant of the fully dynamic APSP algorithm
of Abraham et al.~\cite{AbrahamCK17} one can achieve subquadratic
update bounds for exact dynamic minimum weight cycle.

We will in fact first solve a slightly more general problem 
that we call the \emph{fully dynamic multiple-pairs shortest paths}
(fully dynamic MPSP for short).
Our goal is to have a data structure that maintains distances
$\dist_G(s_i,t_i)$ for some \emph{fixed} (throughout the course of the algorithm)
$k$ source-target pairs $(s_1,t_1),\ldots,(s_k,t_k)$
subject to fully dynamic vertex updates.
Obviously, the classical fully dynamic APSP corresponds to the case $k=n^2$.

In the following we sketch the approach of~\cite{AbrahamCK17} to fully dynamic APSP.
The presentation is however directed towards our goal of obtaining
an MPSP data structure. Some details and proofs can be found only in~\cite{AbrahamCK17};
we focus on the details of our adjustments.

\paragraph{Reduction to batch-deletion MPSP data structure.}
The first step is to reduce the fully dynamic problem to a certain decremental
problem, called the \emph{batch-deletion MPSP}.
In this problem, we want to preprocess the input digraph $G$, so
that one can efficiently compute MPSP in $G\setminus D$ for a subset $D\subseteq V$
that constitutes the query parameter.
We assume that if $G\setminus D$ has a negative cycle, the data structure
has to report its existence instead.
\newcommand{\tprep}{T_\mathrm{pre}}
\newcommand{\tque}{T_\mathrm{q}}

\begin{lemma}\label{l:batch-reduction}
  Suppose we have a batch-deletion MPSP data structure with preprocessing time\linebreak
  $\tprep(n,m,k)$ and worst-case query time $\tque(n,m,k,d)$, where $d=|D|$
  is the size of the removed vertex set.
  Then, for any integer $\Delta>0$, there exists a fully dynamic MPSP algorithm with 
  worst-case update time $O(\tprep(n,m,k)/\Delta+\tque(n,m,k,\Delta)+\Delta (m+k+n\log{n}))$.
\end{lemma}
\begin{proof}[Proof sketch.]
  To obtain an amortized (as opposed to worst-case) bound from the statement, we split the timeline
  into phases of $\Delta$ updates.
  When a new phase starts, we rebuild the batch-deletion data structure from scratch on
  the graph $G_0$ at the start of the phase;
  this clearly incurs $O(\tprep(n,m_0,k)/\Delta)$ amortized time cost per update,
  where $m_0=|E(G_0)|$.
  At some point of a phase, let $D\subseteq V$, $|D|\leq \Delta$, be the vertices touched by updates in this phase.
  To compute MPSP at that point, we first compute
  MPSP in $G_0\setminus D=G\setminus D$ in $O(\tque(n,m_0,k,|D|))=O(\tque(n,m_0,k,\Delta))$ time.
  To obtain MPSP in~$G$, we need to check if paths going through $D$ in $G$
  improve upon those in $G\setminus D$, i.e., we compute MPSP in~$G$ according to
  the equation $\dist_G(s_i,t_i)=\min\left(\dist_{G_0\setminus D}(s_i,t_i),\min_{v\in D}\{\dist_G(s_i,v)+\dist_G(v,t_i)\}\right).$
  
  Observe that all distances of the form $\dist_G(\cdot,v)$ or $\dist_G(v,\cdot)$ for $v\in D$ can be obtained by running Dijkstra's
  algorithm to/from each such $v$, in 
  $O(\Delta(m+n\log{n}))$ total time, as long as a feasible price function of $G$ is given.
  A feasible price function can be maintained in $O(m+n\log{n})$ worst-case time
  after a vertex update using the following theorem proved later in Section~\ref{s:neg-worstcase}.
  \begin{restatable}{theorem}{worstcasenegcycle}\label{t:worst-case-cycle}
  Let $G$ be an initially empty real-weighted digraph. There exists an algorithm maintaining
  the information whether $G$ has a negative cycle and supporting
  vertex updates in $O(m+n\log{n})$ \emph{worst-case} time.
  Additionally, whenever $\mc(G)\geq 0$, the algorithm maintains
  a feasible price function $p$ of $G$.
  \end{restatable}
  Theorem~\ref{t:worst-case-cycle} is also used to keep track of whether the current $G$ has a negative cycle.
  Once the distances from/to $v$ in $G$ are available, the distances $\dist_G(s_i,t_i)$ can be computed in
  $O(\Delta\cdot k)$ time.

    Unfortunately, the above argument is fully valid only if either
  the number of edges $m$ is of the same order throughout, i.e., $m_0=O(m)$, or it cannot drop by more
  than a constant factor during a single phase, e.g.,
  $m=\Omega(n\Delta)$.
  If, say, $\tprep(n,m,k)=\Theta(nm)$, $\Delta=n^{1/3}$ and $m=n^{5/4}=m_0$ at the beginning of the
  phase, and during $n^{1/4}$ first updates in that phase
  $m$ gets decreased to $O(n)$, then the total update cost
  coming from the preprocessing in this
  phase is $\Theta(nm_0)=\Theta(n^{9/4})$.
  If the amortized update time coming from the preprocessing
  was indeed $O(\tprep(n,m,k)/\Delta)$, the the total update cost coming
  from these terms in that phase would be $O(n^{1/4}\cdot nm_0/\Delta+\Delta\cdot n^2/\Delta)=O(n^{13/6})$, which is less by a polynomial factor..

  We circumvent this problem\footnote{This problem does not arise in~\cite{AbrahamCK17}, since there $m$ is assumed
  to be $\Theta(n^2)$ throughout.} as follows. We build the batch-deletion MPSP data structure
  on the graph $G'_0=G_0\setminus D^*$ instead of $G_0$, where $D^*$ is the set
  of $\Delta$ vertices of $G_0$ with highest degrees.
  Then, Dijkstra's algorithm is used to separately compute
  shortest paths through $D\cup D^*$ in~$G$, as opposed to only through $D$.
  Clearly, the cost of such computation remains $O(\Delta(m+k+n\log{n}))$.
  However, the update cost coming from the batch-deletion MPSP
  data structure is decreased to $O(\tprep(n,m_0',k)/\Delta+\tque(n,m_0',k,\Delta))$, where
  $m_0'$ is the number of edges in $G_0'$.
  It is hence enough to observe that $m'_0\leq m$ throughout this phase.
  Indeed, the updates centered at vertices $D$ cannot remove
  more than $\sum_{v\in D}\deg_{G_0}(v)$ edges out of those originally
  contained in $G_0$.
  As a result, $m\geq m_0-\sum_{v\in D}\deg_{G_0}(v)$.
  On the other hand, by removing $D^*$ from $G_0$ we remove at least $\frac{1}{2}\sum_{v\in D^*}\deg_{G_0}(v)$
  edges from $G_0$, i.e, $m_0'\leq m_0-\frac{1}{2}\sum_{v\in D^*}\deg_{G_0}(v)$.
  We obtain $m_0'\leq m$ as follows:
  \begin{equation*}
    m_0'\leq m_0-\frac{1}{2}\sum_{v\in D^*}\deg_{G_0}(v)\leq m_0-\frac{1}{2}\sum_{v\in D}\deg_{G_0}(v)\leq m_0+\frac{1}{2}(m-m_0)=\frac{1}{2}m_0'+\frac{1}{2}m.
  \end{equation*}
  Since the amortization comes only from a (costly) rebuilding step after every
  $\Delta$ updates, turning the amortized bound into a worst-case one is standard, see e.g.,~\cite[Section 2]{AbrahamCK17}.
\end{proof}

\paragraph{The batch-deletion data structure.}
Abraham et al.~\cite{AbrahamCK17} showed a batch-deletion APSP data structure
with $\Ot(n^3)$ preprocessing time and $\Ot(n^2\sqrt{nd})$ query time
which, by Lemma~\ref{l:batch-reduction}, implies $\Ot(n^{2+2/3})$ worst-case
update time for fully dynamic APSP.
Their batch-deletion data structure is Monte Carlo randomized and
produces answers correct with high probability.
We generalize this data structure to MPSP and non-dense graphs.

\begin{theorem}\label{t:batch-deletion}
  There exists a Monte Carlo randomized batch-deletion MPSP data
  structure with $O((m+k)n\log^2{n})$ preprocessing and
  $O((m+n\log{n}+k)\sqrt{nd}\log{n})$
  query time. The answers produced are correct with high probability.
\end{theorem}

Before we prove Theorem~\ref{t:batch-deletion},
let us show how it can be used to obtain fully dynamic
MPSP and minimum weight cycle algorithms.

By choosing $\Delta=n^{1/3}\log^{2/3}n$, and applying
Lemma~\ref{l:batch-reduction}, we obtain:

\tmpsp*
Now consider the fully dynamic minimum weight cycle problem.
The minimum weight of a cycle in $G$
is given by $\mc(G)=\min_{uv\in E(G)}\{\dist_G(v,u)+\wei_G(uv)\}$.
As a result, after each update it is enough to recompute
distances $\dist_G(s_l,t_l)$ in $G$ for $k=m$ pairs $(s_l,t_l)$ such that
$t_ls_l\in E(G)$.
If the edge set of $G$ was fixed (and, for example, the updates were only allowed to
change edge weights), so would be the set of source-target pairs
of our interest. Hence, we could simply use the fully dynamic MPSP data
structure of Theorem~\ref{t:mpsp} in a black-box way.
However, in general, $E(G)$ is not fixed and we need to be more careful.

We proceed as follows. In the reduction of Lemma~\ref{l:batch-reduction}, we will always
build a batch-deletion MPSP data structure with the set of source-target
pairs equal to the edge set used to build that data structure reversed.
This means that at any point of the phase, we can compute
the minimum weight cycle in $G\setminus (D^*\cup D)=G_0\setminus (D^*\cup D)$
in
\begin{equation*}
  O(\tprep(n,m,m)/\Delta+\tque(n,m,m,\Delta))=O((m+n\log{n})n^{2/3}\log^{4/3}n)
\end{equation*}
worst-case time.
Since $G_0\setminus (D^*\cup D)$ contains only a subset of edges of $G_0$,
reading the subset of entries of the distance matrix of $G_0\setminus (D^*\cup D)$
corresponding to reversed edges of $E(G_0)$ is enough to this end.
In order to find the minimum weight cycle going through some vertex of $D^*\cup D$ in $G$,
we just run the algorithm of Observation~\ref{obs:through} (or, more generally,
in presence of negative edges -- 
the algorithm of Lemma~\ref{l:price} with a feasible price function maintained by the algorithm of Theorem~\ref{t:worst-case-cycle})
$|D^*\cup D|=O(\Delta)$ times.
This costs $O(\Delta(m+n\log{n}))=\Ot(mn^{1/3})$ time.

{\renewcommand\footnote[1]{}\exactcycle*}

\subsection{Overview of the batch-deletion MPSP data structure}

Let us now sketch the idea behind our generalization
of the batch-deletion data structure of~\cite{AbrahamCK17}.
The details are given in Section~\ref{s:batch-deletion}.

We first need to refer to some details of the construction of Abraham et al.~\cite{AbrahamCK17}.
The batch-deletion data structure separately handles recomputing
shortest paths of hop-length at least $\sqrt{n/d}$ (``long'' paths), and separately
``short'' shortest paths -- with hop-lengths in the intervals of the form $[h/2,h)$ for $O(\log{n})$ values $h=2^1,2^2,\ldots,\sqrt{n/d}$.

The main difficulty lies in handling short paths, whereas handling long paths
is an easier task.
The key idea (which dates back to Thorup~\cite{Thorup05}) is to compute
an ordered subset $\{v_1,\ldots,v_\ell\}\subseteq V$ with the following properties.
Let $G_i=G\setminus \{v_1,\ldots,v_{i-1}\}$.
Let~$\mathcal{P}_i$ be the set of shortest $\leq h$-hop paths from/to $v_i$ in $G_i$.
Then:
\begin{enumerate}[label=(\arabic*)]
  \item For any $s,t\in V$, an $s\to t$ path not longer than the shortest $\leq h$-hop
    $s\to t$ path in $G$ can be obtained by stitching, for some $i\in\{1,\ldots,\ell\}$, the $s\to v_i$ and $v_i\to t$
    paths from $\mathcal{P}_i$.
  \item For any $x\in V$, $x$ lies on at most $\Ot(hn)$ paths from $\bigcup_{i=1}^\ell\mathcal{P}_{i}$.
\end{enumerate}
Such an ordering, along with the paths $\mathcal{P}_i$, can be computed in $\Ot(nmh)$ time
deterministically (then we have $\ell=n$), or in $\Ot(nm)$ time using randomization (then we can achieve $\ell=\Ot(n/h)$).
Each subsequent vertex $v_i$ in the ordering is picked to be, roughly speaking, the ``most congested'' one
out of $V\setminus\{v_1,\ldots,v_{i-1}\}$,  i.e., the one that has not been picked yet and appears most often on the
previously constructed paths $\bigcup_{j=1}^{i-1}\mathcal{P}_j$.

Given the above, Abraham et al.~\cite{AbrahamCK17} show that after removing any $D\subseteq V$
from $G$, the ``short'' paths in $G$ can be recomputed by:
\begin{enumerate}[label=(\arabic*)]
  \item constructing a number of \emph{sketch graphs} $H_1,\ldots,H_\ell$, where $H_i\subseteq G_i\setminus D$,
  \item rebuilding destroyed (by the removal of $D$) paths from $\mathcal{P}_i$ by running Dijkstra's algorithm from/to $v_i$ on $H_i$,
  \item stitching the reconstructed paths back to obtain paths at leas as good as the actual shortest $\leq h$-hop paths in $G$.
\end{enumerate}
Abraham et al.~\cite{AbrahamCK17} prove that
if we denote by $U_i$ the set of vertices $u$ such that either of the paths $u\to v_i$ or $v_i\to u$
from $\mathcal{P}_i$ has been destroyed by removing $D$, $d=|D|$, then we have $\sum_{i=1}^\ell|U_i|=\Ot(hnd)$,
and the total number of edges $M$ in the sketch graphs is
\begin{equation*}
  M=O\left(\sum_{i=1}^\ell\left(n+\sum_{u\in U_i}\deg_G(u)\right)\right).
\end{equation*}
It is easy to see that $M=\Omega(n\ell)$, and $M=\tilde{O}(hn^2d)$.
Moreover, for each rebuilt path $u\to v_i$ or $v_i\to u$, stitching takes
additional $\Theta(n)$ time -- as one needs to traverse through $\Theta(n)$
source-target pairs that might benefit from this -- 
for a total of $\Ot(hn^2d)$ time.
Since $h$ ranges from $O(1)$ to $\Theta\left(\sqrt{n/d}\right)$, rebuilding short paths
takes $\Omega(n^2)$ and $\Ot(n^2\sqrt{nd})$ time as claimed.

Now, to obtain our improved $\Ot((m+k)\sqrt{nd})$ bound on batch deletion for sparse graphs and small number $k$ of
source-target paths $(s_i,t_i)$ of interest, we make two main adjustments.

First of all, we show that even smaller sketch graphs $H_i$ -- 
with $O\left(\sum_i^\ell\sum_{u\in U_i}\deg_G(u)\right)$ edges in total -- can be used,
thus eliminating the $\Omega(n\ell)$ term, which for small $h$ is $\Omega(n^2)$.

More importantly, we use a different \emph{weighted} scheme for picking
the ordered subset $\{v_1,\ldots,v_\ell\}$. Let us denote by $K$ the undirected
graph on $V$ whose edges correspond to the source-target pairs $(s_i,t_i)$ of interest.
In our scheme, the congestion that a previously computed $\leq h$-hop path
$P=v_i\to u$ (or $P=u\to v_i$) incurs upon some vertex $x$ with $x\in V(P)$ is
$\deg_G(u)+\log{n}+\deg_K(u)$, as opposed to $1$ in~\cite{AbrahamCK17}.
This makes the total congestion of each vertex $x$ in the process
possibly increase to $\tilde{\Theta}(h(m+n\log{n}+k))$, as opposed to $\Ot(hn)$ in~\cite{AbrahamCK17}.
However, we show that the total cost of running Dijkstra's algorithm
on our (more compact) sketch graphs $H_1,\ldots,H_\ell$ can be charged to
the part of the total congestion of removed vertices $D$ coming from the $[\deg_G(u)+\log{n}]$ terms,
which is $\Ot(dh(m+n\log{n}))$.
A similar argument applies to the cost of restitching, which we prove to be
$\Ot(dhk)$.

\subsection{Batch-deletion MPSP data structure}\label{s:batch-deletion}
In this section we provide a proof of Theorem~\ref{t:batch-deletion}.
As our adjustments to the data structure
of Abraham et al.~\cite{AbrahamCK17} are somewhat low-level, this section
also contains a presentation of this data structure using
our notation, with some unaltered proofs deferred to~\cite{AbrahamCK17}.

Let us first assume that $G$ has non-negative edge weights.
We remove this assumption later on.

For a path $P\subseteq G$, let $|P|$ denote the hop-length or $P$.
Moreover, denote by $\dist_G^k(s,t)$ the length of the shortest $\leq k$-hop
$s\to t$ path in $G$. Clearly, we have $\dist_G(s,t)=\dist_G^{n-1}(s,t)$.

Denote by $P_G(s,t)$ the shortest $s\to t$ path in $G$.
We may consider $P_G(s,t)$ uniquely defined by imposing
any fixed (e.g., lexicographical) order on the shortest $s\to t$ paths
in $G$ and defining $P_G(s,t)$ more precisely as the
smallest shortest $s\to t$ path in $G$ according to that order.
We will use the following well-known fact. 

\begin{fact}\label{l:hubs}
  Let $h$ be an integer, $1\leq h\leq n$. 
  With high probability\footnote{Depending on the constant hidden in the $\Theta$ notation.}, a random subset $C\subseteq V$ of size $\Theta((n/h)\log{n})$
  hits all the paths $P_G(s,t)$ for $s,t\in V$ with $|P_G(s,t)|\geq h/2$.
\end{fact}

Let $h=2^\ell$, $1\leq h\leq n$, be an integral parameter to be set later.
Let $h_i=2^i$ for $i=1,\ldots,\ell$.
For $i=1,\ldots,\ell$, let $C_i$ be a random subset of $V$ of
size $\Theta((n/h_i)\log{n})$.
The general idea is to leverage the set $C_i$ to efficiently recompute shortest paths
with hop-lengths in $[h_i/2,h_i)$.
Additionally, we will also use the last set $C_\ell$ to recompute paths with at least
$h=h_\ell$ hops.
By Fact~\ref{l:hubs}, for all~$i$, $C_i$ hits shortest paths
with $h_i$ or more hops in $G$ whp.

What is more important, $C_i$
hits all $\geq h_i/2$-hop shortest paths in graphs $G\setminus D_j$
for polynomially many different sets $D_j\subseteq V$, provided
that the choice of these subsets is independent of $C_i$,
i.e., if the adversary does not know $C_i$.
As a result, we can sample the sets $C_i$ just once during
the preprocessing stage and use them to process $O(\poly{n})$
different batch-deletion queries $D$, while remaining
almost certain that $C_i$ hits the $\geq h_i/2$-hop shortest paths of $G\setminus D$.

\newcommand{\cp}{{\overline{C}}}

For technical reasons that will become clear later, we will need to slightly augment
the sampled sets $C_i$ and subsequently impose an order on them.
So, suppose $\cp_i$ is such that $C_i\subseteq \cp_i$, and $|\cp_i|=\Theta((n/h_i)\log{n})$.
Moreover, let $\cp_i=\{c_{i,1},\ldots,c_{i,|\cp_i|}\}$.
For $j=1,\ldots,|\cp_i|$, let $\cp_{i,j}=\{c_{i,1},\ldots,c_{i,j-1}\}$.
Clearly, $\emptyset=\cp_{i,1}\subset \cp_{i,2}\subset \ldots \subset \cp_{i,|\cp_i|}\subset \cp_i$.

\newcommand{\dlong}{d^\mathrm{long}}
Let $H\subseteq G$ be any subgraph of $G$.
For $i=0,\ldots,\ell$, $s,t\in V$, and $j=1,\ldots,|\cp_i|$, define
\begin{equation}
  d_{H}^{i,j}(s,t)=\dist_{H\setminus \cp_{i,j}}^{h_i}(s,c_{i,j})+\dist_{H\setminus \cp_{i,j}}^{h_i}(c_{i,j},t).
\end{equation}
In words, $d_{H}^{i,j}$ is the minimal length of a $s\to t$ path in $H$ that:
\begin{enumerate}
  \item avoids
all vertices $\cp_{i,j}$,
  \item goes through a vertex $c_{i,j}$ and this
vertex is no more than $h_i$ hops apart from both endpoints on that path.
\end{enumerate}
Let $\dlong_H(s,t)$ be the minimal length of an $s\to t$ path
in $G$ that goes through $C_\ell$:
\begin{equation}
\dlong_H(s,t)=\min_{v\in C_\ell}\left\{\dist_{H}(s_i,v)+\dist_{H}(v,t_i)\right\}.
\end{equation}

\begin{lemma}\label{l:split}
  For any $H\subseteq G$, and $s,t\in V$, with high probability we have:
  \begin{equation*}
    \dist_H(s,t)=\min\left(\dlong_H(s,t),\min_{i,j}\left\{d_H^{i,j}(s,t)\right\}\right).
  \end{equation*}
\end{lemma}
\begin{proof}
  Let $Q=\left\{\dlong_H(s,t)\right\}\cup \bigcup_{i,j}\left\{d_H^{i,j}(s,t)\right\}$.
  We have to prove $\dist_H(s,t)=\min Q$.
  Note that each element constitutes a length of some $s\to t$ path in $H$,
  so $\dist_H(s,t)\leq \min Q$.
  As a result, it is sufficient to prove
  $\dist_H(s,t)\geq \min Q$.

  If $|P_H(s,t)|\geq h$, then by Fact~\ref{l:hubs}, $P_H(s,t)$ is (whp.) hit by a vertex in $C_\ell$.
  As a result, $\dist_H(s,t)=\dist_H(s,v)+\dist_H(v,t)$ for some $v\in C_\ell$.
  This implies $\dist_H(s,t)\geq \dlong_H(s,t)\geq \min Q$.

  Otherwise, let $i\in \{1,\ldots,\ell\}$ be such that $|P_H(s,t)|\in [h_i/2,h_i)$.
  By Fact~\ref{l:hubs}, and since $C_i\subseteq \cp_i$, $P_H(s,t)$ contains a vertex of $\cp_i$
  whp.
  Let $j$ be minimal such that $c_{i,j}\in \cp_i$ lies on~$P_H(s,t)$.
  Then, $P_H(s,t)\subseteq H\setminus \cp_{i,j}$.
  Since $P_H(s,t)$ has no more than $h_i$ hops, $c_{i,j}$
  is no more than $h_i$ hops apart from both endpoints $s,t$.
  As a result, we have:
  \begin{equation*}
    \dist_H(s,t)\geq \dist^h_{H\setminus \cp_{i,j}}(s,c_{i,j})+\dist^h_{H\setminus \cp_{i,j}}(c_{i,j},t)=d_H^{i,j}(s,t)\in Q.
  \end{equation*}
  We conclude $\dist_H(s,t)\geq \min Q$.
\end{proof}

\newcommand{\est}{\tilde{d}}

Given a query set $D\subseteq V$, the data structure will compute
for all pairs $(s_l,t_k)$, $l=1,\ldots,k$:
(1) $\dlong_{G\setminus D}(s_l,t_l)$, and (2)
for all $i=1,\ldots,\ell$, an \emph{estimate}
$\est^i_{G\setminus D}(s_l,t_l)$ such that
\begin{equation}\label{eq:est}
  \dist_{G\setminus D}(s_l,t_l)\leq \est^i_{G\setminus D}(s_l,t_l)\leq \min_{j} \left\{d^{i,j}_{G\setminus D}(s_l,t_l)\right\}.
\end{equation}
The former values handle ``long paths'', whereas the latter -- ``short paths''.
By Lemma~\ref{l:split}, computing these values is enough
to obtain $\dist_{G\setminus D}(s_l,t_l)$ with high probability,
since we have:
\begin{align*}
  \dist_{G\setminus D}(s_l,t_l)&\leq \min\left(\dlong_{G\setminus D}(s,t),\min_{i}\left\{\est^i_{G\setminus D}(s,t)\right\}\right)\\
  &\leq \min\left(\dlong_{G\setminus D}(s,t),\min_{i,j}\left\{d_{G\setminus D}^{i,j}(s,t)\right\}\right)\\
  &=\dist_{G\setminus D}(s_l,t_l).
\end{align*}

\paragraph{Long paths.} Handling long paths requires no preprocessing
apart from sampling $C_\ell$.
In order to recompute $\dlong_{G\setminus D}(s_l,t_l)$
for all pairs $(s_l,t_l)$
we simply run Dijkstra's algorithm on $G\setminus D$ from/to all $v\in C_\ell$.
This takes $O((m+n\log{n})\cdot (n/h)\log{n})$ time.
Afterwards, computing each $\dlong_{G\setminus D}(s_l,t_l)$
takes $O(|C_\ell|)=O((n/h)\log{n})$ additional time.
Thus, we~obtain:
\begin{lemma}\label{l:long}
  Recomputing all $\dlong_{G\setminus D}(s_l,t_l)$ takes 
$O((m+n\log{n}+k)\cdot (n/h)\log{n})$ time.
\end{lemma}

\paragraph{Preprocessing for rebuilding short paths.} Recomputing shorter paths efficiently is a much more involved task
and involves heavy preprocessing of $G$ described below.

\newcommand{\con}{\alpha} 

For each $i=1,\ldots,\ell$, we proceed as follows. We start by initializing
the \emph{congestion} $\con(v)$ of each $v\in V$ to $0$.
We also initialize $\cp_i$ to $\emptyset$.
While $C_i\setminus \cp_i\neq \emptyset$, we extend
$\cp_i$ by a single vertex $c_{i,j}$ at a time.
More precisely, suppose that when such an extension is about to happen,
the current size of $\cp_i$ is $z$ and $\cp_i=\{c_{i,1},\ldots,c_{i,z}\}$.
If the most recently added to $\cp_i$ vertex
was in $C_i$, i.e., $c_{i,z}\in C_i$, we pick $c_{i,z+1}$ to be
the most congested vertex $u\in V\setminus (\cp_i\cup C_i)$,
i.e., the one such that $\con(u)$ is maximum.
Otherwise, if $z=0$ or $c_{i,z}\notin C_i$, we pick $c_{i,z+1}$ to be
an arbitrary vertex of $C_i\setminus \cp_i$.
After picking $c_{i,z}$ and before the subsequent picks happen, we will adjust the
vertices' congestions.
Observe that such a selection process indeed guarantees that $C_i\subseteq \cp_i$ and
$|\cp_i|\leq 2|C_i|=O((n/h_i)\log{n})$.

\newcommand{\pfrom}{\pi^{\mathrm{from}}}
\newcommand{\pto}{\pi^{\mathrm{to}}}

We now describe the preprocessing for each subsequent $c_{i,j}$ and
how the congestions of individual vertices are adjusted.
Recall that when $c_{i,j}$ is known, $\cp_{i,j}$ is defined as well.

The first step is to compute $\leq h_i$-hop shortest paths in $G\setminus \cp_{i,j}$
from/to $c_{i,j}$. Some shortest $\leq h_i$-hop path (if it exists) to each $v\in V\setminus \cp_{i,j}$ from $c_{i,j}$
is stored as $\pto_{i,j}(v)$.
Similarly, some shortest $\leq h_i$-hop from each $v\in V\setminus\cp_{i,j}$ to $c_{i,j}$
is stored as $\pfrom_{i,j}(v)$.
Computing these paths requires running a variant of Bellman-Ford algorithm
up to depth $h_i$ from/to $c_{i,j}$ and thus
takes $O(mh_i)$ time.
Since $|\cp_i|=O((n/h_i)\log{n})$, the total time cost of this step
through all $c_{i,j}\in\cp_i$ is $O(mn\log{n})$.

Let us now define an undirected graph $K$ on the vertex set $V$ that
represents pairs $(s_l,t_l)$, $l=1,\ldots,k$ or our interest.
For each such pair, there is a corresponding edge $s_lt_l$ in $K$.

For each of the paths $\pfrom_{i,j}(v)$ or $\pto_{i,j}(v)$, 
we add $\deg_G(v)+\deg_K(v)+\log{n}$ to the congestion $\con(x)$ of each vertex $x$ lying
on these paths. Since these paths contain $O(h_i)$ vertices, the total
congestion introduced by these paths is $O((\deg_G(v)+\deg_K(v)+\log{n})\cdot h_i)$.
Consequently, the total congestion added through all $v\in V$
when processing $c_{i,j}$
is:\footnote{In~\cite{AbrahamCK17}, a single unit of congestion is added to each vertex on a path.
As a result, the total congestion added is only $O(h_in)$. Weighting the path congestions by the
degree of the source/target is crucial to obtaining a better bound
on the time needed to rebuild these paths when processing a query.} 
\begin{equation*}
  O\left(h_i\sum_{v\in V}\left(\deg_G(v)+\deg_K(v)+\log{n}\right)\right)=O(h_i(m+k+n\log{n})).
\end{equation*}

Recall that the paths are computed in $G\setminus \cp_{i,j}$,
so the congestions of vertices from $\cp_{i,j}$ are
not increased.

\begin{lemma}\label{l:congestion}
  After all vertices $c_{i,j}\in \cp_i$ are processed,
  for all $v\in V$ we have\linebreak 
  $\con(v)=O(h_i(m+n\log{n}+k)\log{n})$.
\end{lemma}
\begin{proof}[Proof sketch.]
  Abraham et al.~\cite{AbrahamCK17} argue that alternating between
  picking the most-congested vertex and an arbitrary vertex
  as the next $c_{i,j}$ is enough to guarantee that all the congestions
  at the end are of order $O(L\log{n})$, where $L$ is the
  total congestion added for a single $c_{i,j}$.
  In particular, it does not matter how the added congestion
  is distributed among the vertices.
  The lemma follows since $L=O((m+n\log{n}+k)h_i)$ in our case.
\end{proof}

For each vertex $u$ we also store a list $\Pi^i_u$ of pointers to paths of the form
$\pfrom_{i,j}(v)$ or $\pto_{i,j}(v)$ containing $u$.
Observe that the total size of the lists $\Pi^i_u$ is
$O(|\cp_i|\cdot n\cdot h_i)=O(n^2\log{n})$.

Finally, for each source-target pair $(s_l,t_l)$ the pairs $(d^{i,j}_G(s_l,t_l), j)$,
$j=1,\ldots,|\cp_i|$, are stored in a sorted array $S_i(s_l,t_l)$.
Since computing shortest $\leq h_i$-hop paths from/to $c_{i,j}$ also
gives the lengths $\dist_{G\setminus \cp_{i,j}}^{h_i}(c_{i,j},\cdot)$ and $\dist_{G\setminus \cp_{i,j}}^{h_i}(\cdot,c_{i,j})$,
we can also easily compute $d^{i,j}_G(s_l,t_l)$ for each
pair $(s_l,t_l)$ in $O(|\cp_i|)$ time.
Therefore, computing all the required sets $S_i(s_l,t_l)$
and sorting them takes $O(|\cp_i|\cdot k\log{n}))=O((n/h_i)k\log^2{n})$ time.

\begin{lemma}\label{l:short-prep}
  The described preprocessing (through all $i$) takes $O((m+k)n\log^2{n})$ time.
\end{lemma}
\begin{proof}
  The total running time of the described preprocessing is:
  \begin{equation*}
    O\left(\sum_{i=1}^\ell mn\log{n}+(n/h_i)k\log^2{n}\right)=O\left(mn\ell\log{n}+nk\log^2{n}\cdot \sum_{i=1}^\infty \frac{1}{2^i}\right),
  \end{equation*}
  which yields the desired bound since $\ell=O(\log{h})=O(\log{n})$.
\end{proof}

\paragraph{Rebuilding short paths upon query.}
Recall that upon query $D$, $|D|=d$, our goal is to compute estimates
$\est^i_{G\setminus D}(s_l,t_l)$ satisfying~\eqref{eq:est}.

For a fixed $i$, let $X_i(s_l,t_l)$ be the set of those
$j$ such that $D$ intersects either of the paths $\pfrom_{i,j}(s_l)$
or $\pto_{i,j}(t_l)$.

\begin{lemma}\label{l:kcharge}
  We have $\sum_{l=1}^k |X_i(s_l,t_l)|=O(dh_i(m+n\log{n}+k)\log{n})$.
  Moreover, all the sets $X_i(s_l,t_l)$ for $l=1,\ldots,k$ can be computed
  in $O(dh_i(m+n\log{n}+k)\log{n})$ time.
\end{lemma}
\begin{proof}
For each $v\in D$, we iterate through all the paths $\pi$ of the list $\Pi^i_v$.
If $\pi=\pfrom_{i,j}(s_l)$ for some $i,j,l$, we go through all the
neighbors $t'$ of $s_l$ in $K$.
If $s_lt'$ corresponds to some
  pair $(s_{l'},t_{l'})$, we add $j$ to $X_i(s_{l'},t_{l'})$.
  Similarly, if $\pi=\pto_{i,j}(t_l)$, we go through all the neighbors
  $s'$ of $t_l$ in $K$.
  If $s't_l$ corresponds to some
  pair $(s_{l'},t_{l'})$, we add $j$ to $X_i(s_{l'},t_{l'})$.

  Let us analyze the running time of this algorithm limited
  to handling paths of the form $\pi=\pfrom_{i,j}(s_l)$.
  The analysis in the other case is analogous.
  For a fixed $v\in D$,
  and each $\pfrom_{i,j}(s_l)$ with $v\in V(\pfrom_{i,j}(s_l))$,
  we spend $O(\deg_K(s_l))$ time
  iterating through the neighbors of $s_l$ in the graph $K$.
  We charge this cost to the $\deg_K(s_l)$ contribution
  of the path $\pfrom_{i,j}(s_l)$ to the congestion $\con(v)$.
  Note that some of the costs -- in case $\pfrom_{i,j}(s_l)$
  has many vertices from $D$ -- can be charged to multiple vertices
  of $D$; however, what matters is that no part of any $\con(v)$ is used to ``pay''
  for two distinct neighbors iterations.
  As a result, the total time can be bounded
  as $O\left(\sum_{v\in D}\con(v)\right)=O(dh_i(m+n\log{n}+k)\log{n})$.
  
  Since the total size of all the sets $X_i(s_l,t_l)$ cannot
  be larger than the time needed to construct these
  sets, the above asymptotic expression also bounds $\sum_{l=1}^k |X_i(s_l,t_l)|$.
\end{proof}

  Recall that our goal is to compute, for each $i$ and $l$, such an estimate $\est^i_{G\setminus D}(s_l,t_l)$, that:
\begin{align*}
  \dist_{G\setminus D}(s_l,t_l)\leq \est^i_{G\setminus D}(s_l,t_l)&\leq \min_{j}\left\{d^{i,j}_{G\setminus D}(s_l,t_l)\right\}\\
                  &=\min\left(\min_{j\in X_i(s_l,t_l)}\left\{d^{i,j}_{G\setminus D}(s_l,t_l)\right\},\min_{j\notin X_i(s_l,t_l)}\left\{d^{i,j}_{G\setminus D}(s_l,t_l)\right\}\right).
\end{align*}
Therefore, for each $l$ we will separately compute:
\begin{enumerate}[label=(\arabic*)]
  \item the value $\min_{j\notin X_i(s_l,t_l)}\left\{d^{i,j}_{G\setminus D}(s_l,t_l)\right\}$,
  \item the length of some $s_l\to t_l$ path in $G\setminus D$ not longer than
$\min_{j\in X_i(s_l,t_l)}\left\{d^{i,j}_{G\setminus D}(s_l,t_l)\right\}$.
\end{enumerate}
We will then use the minimum of the former and the latter as $\est^i_{G\setminus D}(s_l,t_l)$.

Consider item~(1).
Observe that $j\notin X_i(s_l,t_l)$ implies that an $s_l\to t_l$ 
path \linebreak $\pfrom_{i,j}(s_l)\cdot \pto_{i,j}(t_l)\subseteq G$ exists also
in $G\setminus D$.
Hence,  $d^{i,j}_G(s_l,t_l)=d^{i,j}_{G\setminus D}(s_l,t_l)$,
and thus $(d^{i,j}_{G\setminus D}(s_l,t_l),j)\in S_i(s_l,t_l)$.
As a result, the minimum
$\min_{j\notin X_i(s_l,t_l)}\left\{d^{i,j}_{G}(s_l,t_l)\right\}$
can be computed in
$O(|X_i(s_l,t_l)|)$ time by inspecting at most $|X_i(s_l,t_l)|+1$ leading
elements of the sorted array $S_i(s_l,t_l)$: for the first encountered element
$(x,j)$ with $j\notin X_i(s_l,t_l)$,~$x$ is the sought minimum.
Note that, by Lemma~\ref{l:kcharge}, computing such minima through all $l$
takes time
\begin{equation*}
  O\left(\sum_{l=1}^k |X_i(s_l,t_l)|\right)=O(dh_i(m+n\log{n}+k)\log{n}).
\end{equation*}

\newcommand{\hfrom}{H^{\mathrm{from}}}
\newcommand{\hto}{H^{\mathrm{to}}}

Now consider~item~(2). For each $c_{i,j}\in \cp_i$, we construct
two auxiliary graphs $\hto_{i,j}$, $\hfrom_{i,j}$. Let us define the former;
the latter is defined completely symmetrically.

The vertices of $\hto_{i,j}$ are precisely
those $z\in V$ satisfying
$D\cap V(\pto_{i,j}(z))\neq \emptyset$, and the vertex $c_{i,j}$.
For each edge $vz\in E(G\setminus D)$
with $v,z\in V(\hto_{i,j})$, we add that edge $vz$ to $\hto_{i,j}$.
Otherwise, if $v\notin V(\hto_{i,j})$ but
$z\in V(\hto_{i,j})$, we add an edge $c_{i,j}z$
with weight $\dist^{h_i}_{G\setminus \cp_{i,j}}(c_{i,j},v)+\wei_{G\setminus D}(vz)$.
Recall that the value $\dist^{h_i}_{G\setminus \cp_{i,j}}(c_{i,j},v)=\len(\pto_{i,j}(v))=\dist^{h_i}_{G\setminus \cp_{i,j}\setminus D}(c_{i,j},v)$
was computed during preprocessing.

\begin{lemma}\label{l:sketch}
  For any $l$ such that $j\in X_i(s_l,t_l)$:
  \begin{itemize}
    \item If $D\cap V(\pto_{i,j}(t_l))\neq\emptyset$, then
    $\dist_{G\setminus D}(c_{i,j},t_l)\leq \dist_{\hto_{i,j}}(c_{i,j},t_l)\leq  \dist_{G\setminus \cp_{i,j}\setminus D}^{h_i}(c_{i,j},t_l)$.
  \item If $D\cap V(\pfrom_{i,j}(s_l))\neq\emptyset$ then, $\dist_{G\setminus D}(s_l,c_{i,j})\leq \dist_{\hfrom_{i,j}}(s_l,c_{i,j})\leq  \dist_{G\setminus \cp_{i,j}\setminus D}^{h_i}(s_l,c_{i,j})$.
  \end{itemize}
  \end{lemma}
\begin{proof}
  We only prove the former claim; one can prove the latter by proceeding symmetrically.
  If $c_{i,j}\in D$, the claim is trivial,
  so assume $c_{i,j}\notin D$.
  First of all, $t_l$ is a vertex of $\hto_{i,j}$ by construction and
  the assumption $D\cap V(\pto_{i,j}(t_l))\neq\emptyset$.
  The desired lower bound on $\dist_{\hto_{i,j}}(c_{i,j},t_l)$ holds since
  all edges of $\hto_{i,j}$ are either also edges of $G\setminus D$,
  or encode lengths of paths that exist in $G\setminus D$.

  To obtain the upper bound, let path $P$ be some shortest
  among $\leq h_i$-hop $c_{i,j}\to t_l$ paths in $G\setminus \cp_{i,j}\setminus D$.
  Let $v$ be the last vertex on $P$ such that $D\cap V(\pto_{i,j}(v))=\emptyset$.
  Note that $v\neq t_l$ exists since $\pto_{i,j}(c_{i,j})$ is a zero-hop path and $c_{i,j}\notin D$.
  Let $y$ be the vertex following $v$ on $P$.
  Split $P=P_1P_2$ so that~$P_2$ is a $y\to t_l$ path.
  Note that $V(P_2)\subseteq V(\hto_{i,j})$.
  Since all edges between the subset $V(\hto_{i,j})$ in $G\setminus D$
  are preserved in $\hto_{i,j}$, we have $P_2\subseteq \hto_{i,j}$.
  Moreover, since 
  $D\cap V(\pto_{i,j}(v))=\emptyset$
  and $v$ is a neighbor of a vertex $y\in V(\hto_{i,j})$,
  there is an edge $c_{i,j}y$ in $\hto_{i,j}$
  of weight $\dist^h_{G\setminus\cp_{i,j}}(c_{i,j},v)+\wei_{G\setminus D}(vy)=\dist^h_{G\setminus\cp_{i,j}\setminus D}(c_{i,j},v)+\wei_{P_1}(vy)\leq \len(P_1)$.
  As a result, we have
  $\dist_{\hto_{i,j}}(c_{i,j},t_l)\leq \len((c_{i,j}w)\cdot P_2)\leq \len(P)=\dist_{G\setminus \cp_{i,j}\setminus D}^{h_i}(c_{i,j},t_l)$.
\end{proof}

For a fixed $i$, in order to find -- for all $l$ -- a path of length no more than\linebreak
$\min_{j\in X_i(s_l,t_l)}\left\{d^{i,j}_{G\setminus D}(s_l,t_l)\right\}$,
we proceed as follows.
For each $j$, we build the graphs $\hto_{i,j}$ and $\hfrom_{i,j}$.
We run Dijkstra's algorithm from $c_{i,j}$ in $\hto_{i,j}$,
and Dijktra's algorithm from $c_{i,j}$ in the reverse of $\hfrom_{i,j}$.
This way, for all $v\in V(\hto_{i,j})$, we obtain $\dist_{\hto_{i,j}}(c_{i,j},v)$,
and for each $v\in V(\hfrom_{i,j})$, we obtain $\dist_{\hfrom_{i,j}}(v,c_{i,j})$.
Finally, for each $l$ we iterate through all
$j\in X_i(s_l,t_l)$ in order to find the shortest among the candidate $s_l\to t_l$ paths
through some $c_{i,j}$.
The length of such a candidate path through $c_{i,j}$ is $\lambda_{i,j}(s_l,t_l)$,
defined as:
\begin{equation*}
  \lambda_{i,j}(s_l,t_l)=\begin{cases}
    \dist_{\hfrom_{i,j}}(s_l,c_{i,j})+\dist_{\hto_{i,j}}(c_{i,j},t_l) & \text{ if }D\cap\pfrom_{i,j}(s_l)\neq\emptyset\text{ and } D\cap\pto_{i,j}(t_l)\neq\emptyset,\\
    \dist_{G\setminus \cp_{i,j}}^{h_i}(s_l,c_{i,j})+\dist_{\hto_{i,j}}(c_{i,j},t_l) & \text{ if }D\cap\pfrom_{i,j}(s_l)=\emptyset\text{ and } D\cap\pto_{i,j}(t_l)\neq\emptyset,\\
    \dist_{\hfrom_{i,j}}(s_l,c_{i,j})+\dist_{G\setminus\cp_{i,j}}^{h_i}(c_{i,j},t_l) & \text{ if }D\cap\pfrom_{i,j}(s_l)\neq\emptyset\text{ and } D\cap\pto_{i,j}(t_l)=\emptyset.\\
  \end{cases}
\end{equation*}
Recall that one of these three cases apply by $j\in X_i(s_l,t_l)$.
By Lemma~\ref{l:sketch}, we easily obtain:
\begin{equation*}
\lambda_{i,j}(s_l,t_l)\leq \dist_{G\setminus \cp_{i,j}\setminus D}^{h_i}(s_l,c_{i,j})+\dist_{G\setminus\cp_{i,j}\setminus D}^{h_i}(c_{i,j},t_l)= d^{i,j}_{G\setminus D}(s_l,t_l),
\end{equation*}
and consequently $\min_{j\in X_i(s_l,t_l)}\{\lambda_{i,j}(s_l,t_l)\}\leq \min_{j\in X_i(s_l,t_l)}\{d^{i,j}_{G\setminus D}(s_l,t_l)\}$ as desired.

Let us now bound the running time of rebuilding short paths.
For a given $i$, given distances from/to $c_{i,j}$ in $\hto_{i,j}$ and $\hfrom_{i,j}$ for all $j$,
computing
all the values $\lambda_{i,j}(s_l,t_l)$ clearly takes
$O\left(\sum_{l=1}^k |X_i(s_l,t_l)|\right)$ time, i.e., $O(dh_i(m+n\log{n}+k)\log{n})$ time by Lemma~\ref{l:kcharge}.
It remains to bound the total time needed to computing those
distances in $\hto_{i,j}$ and $\hfrom_{i,j}$.

\begin{lemma}\label{l:mcharge}
  Computing single source/target distances $\dist_{\hto_{i,j}}(c_{i,j},\cdot)$ and $\dist_{\hfrom_{i,j}}(\cdot,c_{i,j})$
  for all $j=1,\ldots,|\cp_i|$ takes $O(dh_i(m+n\log{n}+k)\log{n})$ time.
\end{lemma}
\begin{proof}
  We only consider computing single-source distances in graphs $\hto_{i,j}$; finding distances
  in graphs $\hfrom_{i,j}$ is analyzed analogously.
  The cost of running Dijkstra's algorithm on $\hto_{i,j}$
  is: 
  \begin{equation*}
    O(|E(\hto_{i,j})|+|V(\hto_{i,j})|\log{n})=O\left(\sum_{v\in V(\hto_{i,j})}\left( \deg_{G}(v)+\log{n}\right)\right).
  \end{equation*}
  Recall that we have $v\in V(\hto_{i,j})$ if and only if $D\cap V(\pto_{i,j}(v))\neq\emptyset$.
  We can thus charge the cost $O(\deg_G(v)+\log{n})$ of processing
  $v$ in the corresponding Dijkstra run to the $\deg_G(v)+\log{n}$ (out of
  $\deg_G(v)+\deg_K(v)+\log{n}$; recall that $\deg_K(v)$ part was
  already charged in Lemma~\ref{l:kcharge}) contribution
  of $\pto_{i,j}(v)$ to $\alpha(x)$ for an arbitrary vertex $x\in D\cap V(\pto_{i,j}(v))$.
  The total amount of congestion charged is again $O\left(\sum_{x\in D}\alpha(x)\right)=O(dh_i(m+n\log{n}+k)\log{n})$.  
\end{proof}

By combining Lemmas~\ref{l:short-prep},~\ref{l:kcharge}~and~\ref{l:mcharge}, we obtain:
\begin{corollary}
  After $O((m+k)n\log^2{n})$ preprocessing, computing estimates $\est_{G\setminus D}^i(s_l,t_l)$
  as specified in~\eqref{eq:est} for any $D\subseteq V$, $d=|D|$, takes $O(dh(m+n\log{n}+k)\log{n})$ time.
\end{corollary}

And finally, by combining the above corollary with Lemma~\ref{l:long}, and choosing
$h=\sqrt{n/d}$, we obtain:
\begin{lemma}
  After $O((m+k)n\log^2{n})$ preprocessing, for any $D\subseteq V$, $d=|D|$, one can recompute
  distances $\dist_{G\setminus D}(s_l,t_l)$ for all $l=1,\ldots,k$
  in $O\left(\sqrt{nd}\left(m+n\log{n}+k\right)\log{n}\right)$ time.
\end{lemma}

\paragraph{Dealing with negative edges and cycles.}
Note that the preprocessing computes limited-hop shortest paths
using the Bellman-Ford algorithm.
As a result, negative edges or cycles have no effect on the preprocessing.
However, when handling a query, we repeatedly run Dijkstra's algorithm.
In general, it needs a feasible price function on $G\setminus D$ to work.
To obtain such a feasible price function at query time,
during preprocessing
we set up a fully dynamic negative cycle detection data structure
of Theorem~\ref{t:worst-case-cycle} (see Section~\ref{s:neg-worstcase}) and issue~$n$
vertex updates
to that data structure so that after the preprocessing finishes, that data structure stores
the graph $G$.
Upon query, we remove the vertices $D$ from that data structure,
so that in $O(d(m+n\log{n}))$ worst-case time we either declare that
$G\setminus D$ contains a negative cycle (and thus we do not have to compute MPSP),
or obtain a feasible price function $p$ on $G\setminus D$.
After obtaining $p$, we revert these removals, again in $O(d(m+n\log{n}))$ worst-case
time, so that the data structure again stores the graph $G$.
Note that the extra overhead needed to handle negative edges and cycles
is negligible since $d\leq \sqrt{nd}$.

\subsection{Maintaining a negative cycle in $O(m+n\log{n})$ worst-case time}\label{s:neg-worstcase}
In this section we prove the following theorem.

\worstcasenegcycle*

\newcommand{\tin}{\textrm{in}}
\newcommand{\tout}{\textrm{out}}

Consider the graph $G'$ where each vertex $v$ is split into two vertices
$v_\tin,v_\tout$, so that:
\begin{enumerate}
  \item for each $v$, there is an edge $v_\tin v_\tout$
of weight $0$ in $G'$,
  \item for each edge $uv\in E(G)$, there is an edge
$u_\tout v_\tin$ of weight $\wei_G(uv)$ in $G'$.
\end{enumerate}
Clearly, $G$ has a negative cycle iff $G'$ has a negative cycle.
Moreover, a vertex update to $G$ corresponds to at most two vertex
updates to $G'$.

The key idea is to view $G'$ as a unit-capacitated flow network
with edge costs given by the edge weights.
For some flow $f$ in this network (i.e., any function $f:E(G)\to \mathbb{R}_{\geq 0}$
such that for all $e\in E(G)$, $f(e)\leq w_G(e)$), denote by $c(f)$ its cost.
Recall that a flow $f$ is called a circulation if
for all $v\in V$ we have $\sum_{xv\in E(G)}f(xv)=\sum_{vy\in E(G)}f(vy)$.
\begin{lemma}\label{l:circ-equiv}
  Let $f^*$ be the minimum cost circulation in $G'$. Then $c(f^*)\leq 0$.
  Moreover,  $G'$ has a negative cycle if and only if $c(f^*)<0$.
\end{lemma}
\begin{proof}
  Clearly, since there are no demands at vertices, a zero circulation (i.e., where
  the flow on each edge is $0$) is a valid circulation of cost $0$.
  As a result, we always have $c(f^*)\leq 0$.

  Suppose $G'$ has a negative cycle. Then, by sending a unit of flow through that
  cycle, we obtain a valid circulation of negative cost, i.e., we conclude $c(f^*)<0$.

  Now suppose that $c(f^*)<0$. It is well-known (see e.g.~\cite{networkflows}) that any circulation
  -- in particular $f^*$ --
  can be decomposed into edge-disjoint cycles in $G'$ of positive flow.
  Since the sum of costs of these cycles is negative, at least one of them
  has to have a negative cost.
\end{proof}

Let $G'_f$ be the residual network corresponding to flow $f$.
Recall that in $G'_f$, for each edge $e=uv\in E(G')$,
we have an edge $uv$ of cost $c(uv)=\wei_{G'}(uv)$ in $G'_f$ if $f(e)<1$,
and we have an edge $vu$ of cost $c(vu)=-\wei_{G'}(uv)$ in $G'_f$ if $f(e)>0$.
When working with unit capacities, we can safely operate
on integer flows, i.e., the flow $f(e)$ on each edge $e\in E(G')$ is always either $0$ or $1$.
The following characterization of a minimum cost flow is well known (see e.g.~\cite{networkflows}):
\begin{fact}\label{f:flow}
  A flow $f$ in $G'$ has minimum cost if there exists such a potential function
  $\pi:V(G')\to\mathbb{R}$ that 
  for each edge $e=uv\in E(G'_f)$ we have $c(e)-\pi(u)+\pi(v)\geq 0$.
\end{fact}

Our algorithm will maintain a minimum cost circulation $f^*$ in $G'$,
along with a potential function~$\pi$ as described in Fact~\ref{f:flow} that certifies the optimality
of $f^*$.
Initially, when $G$ is empty, we use a zero circulation and a zero potential function
which clearly works since the graph $G'$ contains only edges $v_\tin v_\tout$ of zero cost.

Suppose a vertex update to $G$ centered at $v$ happens.
First, we cancel any flow that goes through edges incident
to either $v_\tin$ or $v_\tout$ in $f^*$.
This may create an excess of $1$ on some vertex $u_\tout$
such that we had $f^*(u_\tout v_\tin)=1$ before,
and a deficit of $1$ on some vertex $z_\tin$
such that we had $f^*(v_\tout z_\tin)=1$ before.
We then reflect the vertex update to $G$ in $G'$.
Note that afterwards, $\pi$ is still a an \emph{almost} valid
potential function (as in Fact~\ref{f:flow}) -- perhaps except on the edges
of $G_{f^*}'$ incident to either $v_\tin$ or $v_\tout$.
To fix this, we set $\pi(v_\tin)$ to a sufficiently
large value
so that for all edges $u_\tout v_\tin$ (which do not carry flow at this point)
we have $\wei_{G'}(u_\tout v_\tin)-\pi(u_\tout)+\pi(v_\tin)\geq 0$.
Similarly, we set $\pi(v_\tout)$ to a sufficiently
small value so that for all edges
$v_\tout z_\tin$ we have $\wei_{G'}(v_\tout z_\tin)-\pi(v_\tout)+\pi(z_\tin)\geq 0$.
At this point, $v_\tin v_\tout$ is the only
edge in $G_{f^*}'$ that can potentially have
$\wei_{G'}(v_\tin v_\tout)-\pi(v_\tin)+\pi(v_\tout)<0$.
If this is the case, we send flow
through the edge $v_\tin v_\tout$.
At this point, by Fact~\ref{f:flow}, $f^*$ is a minimum cost flow,
albeit it might have at most two excess and at most two deficit vertices.

It is well-known (see e.g.~\cite{KarczmarzS19} for a detailed description)
that one can convert any minimum-cost flow into a
minimum cost circulation by gradually pushing flow from
the excess to deficit vertices in the residual network
through shortest (in terms of cost in the residual network) such paths.
The potential function certifying the optimality
of the flow allows to find shortest paths in the residual network
using Dijkstra's algorithm.
Moreover, the \emph{distances to} the deficit vertices
computed by Dijstra's algorithm can be used as the new potential
function after eliminating a single unit of excess.
Since each such step removes a single unit of excess,
in our case we need at most two such steps.

As a result, after a vertex update we can compute a new
minimum cost circulation and the corresponding potential function $\pi$
using $O(1)$ Dijkstra's algorithm runs. This takes $O(m+n\log{n})$ worst-case time.

Finally, we show how to obtain, based on $f^*$ and $\pi$, a feasible price function $p$
of $G$ if $G$ has no negative cycles.

\begin{lemma}
  Suppose $G$ has no negative cycles. Let $f^*$ be the minimum cost circulation
  in $G'$ and let $\pi$ be the potential function certifying the optimality of $f^*$.
  Let $p$ be such that $p(v)=-\pi(v_\tin)$ for all $v\in V$.
  Then, $p$ is a feasible price function of $G$.
\end{lemma}
\begin{proof}
  First of all, we have $c(f^*)=0$ by Lemma~\ref{l:circ-equiv}.
  As a result, $f^*$ can be decomposed into a collection of $j\geq 0$
  edge-disjoint cycles $C_1,\ldots,C_j$,
  of positive flow.
  In fact, each of these cycles has to have $0$ cost, as otherwise at least
  one of them would have negative cost.

  Consider some of the cycles $C_i$. 
  Consider any edge $uv\in C_i$. Clearly, since $f^*(uv)=1$, 
  $vu\in G_{f^*}'$ and thus $-\wei_{G'}(uv)-\pi(v)+\pi(u)\geq 0$
  by the optimality of $f^*$.
  Equivalently, \linebreak ${\wei_{G'}(uv)-\pi(u)+\pi(v)\leq 0}$.
  Now suppose we have $\wei_{G'}(xy)-\pi(x)+\pi(y)<0$ for some edge $xy\in C_i$.
  Then the sum
  \begin{equation*}
    \sum_{uv\in C_i}(\wei_{G'}(uv)-\pi(v)+\pi(u))
  \end{equation*}
  has non-positive
  terms, and at least one negative term, so in fact it is negative.
  But since $C_i$ is cycle, the potentials $\pi(\cdot)$
  in the above sum cancel out, so we actually obtain
  $\sum_{uv\in C_i}\wei_{G'}(uv)<0$ which contradicts that $C_i$ has zero cost.
  Therefore, we conclude that for any edge $uv\in C_i$ we have:
  \begin{equation}\label{eq:tight}
    \wei_{G'}(uv)-\pi(u)+\pi(v)=0.
  \end{equation}
  Suppose that for some $v$ we have $f^*(v_\tin v_\tout)=1$. Then $v_\tin v_\tout$
  lies on some cycle $C_i$ of the decomposition and thus by~\eqref{eq:tight} we have:
  \begin{equation*}
  \wei_{G'}(v_\tin v_\tout)-\pi(v_\tin)+\pi(v_\tout)=-\pi(v_\tin)+\pi(v_\tout)=0.
  \end{equation*}
  As a result $\pi(v_\tin)=\pi(v_\tout)$ in this case.
  
  If $f^*(v_\tin v_\tout)=0$, then $v_\tin v_\tout\in G_{f^*}'$, so we easily conclude $\pi(v_\tout)\geq \pi(v_\tin)$
  from the optimality of $f^*$.

  Now consider some original edge $uv$ of $G$. There are two cases.
  If $f^*(u_\tout v_\tin)=1$, then
  $u_\tout v_\tin$ lies on some cycle $C_i$ and therefore by~\eqref{eq:tight} we know that
  \begin{equation*}
    \wei_{G'}(u_\tout v_\tin)-\pi(u_\tout)+\pi(v_\tin)=0.
  \end{equation*}
  
  Observe that $u_\tin u_\tout$ also necessarily lies on that cycle, so $\pi(u_\tout)=\pi(u_\tin)$
  and we in fact have:
  \begin{equation*}
    \wei_G(uv)+p(u)-p(v)=\wei_{G'}(u_\tout v_\tin)-\pi(u_\tin)+\pi(v_\tin)=0,
  \end{equation*}
  so in particular $\wei_G(uv)+p(u)-p(v)\geq 0$ as desired.

  Now assume $f^*(u_\tout v_\tin)=0$. Then $u_\tout v_\tin\in G_{f^*}'$ and by $\pi(v_\tout)\geq \pi(v_\tin)$
  we have:
  \begin{equation*}
    \wei_G(uv)+p(u)-p(v)=\wei_{G'}(u_\tout v_\tin)-\pi(u_\tin)+\pi(v_\tin)\geq \wei_{G'}(u_\tout v_\tin)-\pi(u_\tout)+\pi(v_\tin)\geq 0.
  \end{equation*}
  This concludes the proof that $p(v)=-\pi(v_\tin)$ is a feasible price function of $G$.
\end{proof}

\clearpage

\bibliographystyle{plainurl}
\bibliography{references}

\end{document}